\documentclass[journal]{IEEEtran}

\usepackage{amsmath,amsthm,amssymb,bm,booktabs}
\usepackage{algorithm}
\usepackage[noadjust]{cite}
\usepackage[noend]{algpseudocode}
\usepackage{url}

\ifCLASSINFOpdf
  \usepackage[pdftex]{graphicx}
\else
 \usepackage[dvips]{graphicx}
\fi

\newtheorem{thm}{Theorem}
\newtheorem{lem}{Lemma}
\newtheorem{dfn}{Definition}

\begin{document}

\title{Predicting Online Item-choice Behavior:\\ 
A Shape-restricted Regression Approach}
\author{
Naoki~Nishimura,~Noriyoshi~Sukegawa,~Yuichi~Takano,~and~Jiro~Iwanaga
\thanks{N. Nishimura was with the Graduate School of Systems and Information Engineering, 
University of Tsukuba, Ibaraki 305-8577, Japan. 
He is now with Recruit Lifestyle Co., Ltd., Tokyo 100-6640, Japan 
(e-mail: nishimura@r.recruit.co.jp).}
\thanks{N. Sukegawa is with the Department of Information and Computer Technology, 
Tokyo University of Science, Tokyo 125-8585, Japan 
(e-mail: sukegawa@rs.tus.ac.jp).}
\thanks{Y. Takano is with the Faculty of Engineering, Information and Systems, 
University of Tsukuba, Ibaraki 305-8577, Japan 
(e-mail: ytakano@sk.tsukuba.ac.jp).}
\thanks{J. Iwanaga is with 
the Doctoral Program in
Policy and Planning Sciences, 
University of Tsukuba, Ibaraki 305-8577, Japan, 
and also with Erdos Inc., Kanagawa 222-0033, Japan 
(email: iwanaga@erdos-the-book.com).}
}

\markboth{}
{Nishimura \MakeLowercase{\textit{et al.}}: 
Predicting Online Item-choice Behavior: A Shape-restricted Regression Approach}

\maketitle

\begin{abstract}
This paper examines the relationship between user pageview (PV) histories and their item-choice behavior on an e-commerce website. 
We focus on PV sequences, which represent time series of the number of PVs for each user--item pair. 
We propose a shape-restricted optimization model that accurately estimates item-choice probabilities for all possible PV sequences. 
This model imposes monotonicity constraints on item-choice probabilities by exploiting partial orders for PV sequences, according to the recency and frequency of a user's previous PVs. 
To improve the computational efficiency of our optimization model, we devise efficient algorithms for eliminating all redundant constraints according to the transitivity of the partial orders. 
Experimental results using real-world clickstream data demonstrate that our method achieves higher prediction performance than that of a state-of-the-art optimization model and common machine learning methods. 
\end{abstract}

\begin{IEEEkeywords}
Consumer behavior, 
electronic commerce, 
graph theory,
predictive models, 
quadratic programming
\end{IEEEkeywords}

\IEEEpeerreviewmaketitle

\section{Introduction}
\label{sec:1}
\IEEEPARstart{A}{} growing number of companies are now operating e-commerce websites that allow users to browse and purchase a variety of items~\cite{TuOu17}.
Within this context, there is great potential value in analyzing users' item-choice behavior from clickstream data, which is a record of user pageview (PV) histories on an e-commerce website.
By grasping users' purchase intention as revealed by PV histories, we can lead users to target pages or design special sales promotions, providing companies with opportunities to build profitable relationships with users~\cite{Ka17,NgXi09}. 
Companies can also use clickstream data to improve the quality of operational forecasting and inventory management~\cite{HuMi14}.
Meanwhile, users often find it difficult to select an appropriate item from among the plethora of choices presented by e-commerce websites~\cite{Ag16}. 
Analyzing item-choice behavior can improve the performance of recommendation systems that help users discover items of interest~\cite{IwNi19}. 
For these reasons, a number of prior studies have investigated clickstream data from various perspectives~\cite{BuSi09}. 
In this study, we focused on closely examining the relationship between PV histories and item-choice behavior on an e-commerce website.

It has been demonstrated that the recency and frequency of a user's past purchases are critical indicators for purchase prediction~\cite{FaHa05a,VaBu05} and sequential pattern mining~\cite{ChKu09}. 
Accordingly, Iwanaga \textit{et al.}~\cite{IwNi16} developed a shape-restricted optimization model for estimating item-choice probabilities from the recency and frequency of a user's previous PVs. 
Their method creates a two-dimensional probability table consisting of item-choice probabilities for all recency--frequency combinations in a user's previous PVs. 
Nishimura \textit{et al.}~\cite{NiSu18} employed latent-class modeling to integrate item heterogeneity into a two-dimensional probability table. 
These prior studies demonstrated experimentally that higher prediction performance was achieved with the two-dimensional probability table than with common machine learning methods, namely, logistic regression, kernel-based support vector machines, artificial neural networks, and random forests. 
Notably, however, reducing PV histories to two dimensions (recency and frequency) can markedly decrease the amount of information contained in PV histories reflecting item-choice behavior. 

This study focused on PV sequences, which represent time series of the number of PVs for each user--item pair. 
In contrast to the two-dimensional probability table, PV sequences allow us to retain detailed information contained in the PV history. 
However, the huge number of possible PV sequences makes it extremely difficult to accurately estimate item-choice probabilities for all of them.
To overcome this difficulty, we propose a shape-restricted optimization model that imposes monotonicity constraints on item-choice probabilities based on a partially ordered set (poset) for PV sequences. 
While this optimization model contains a huge number of constraints, all redundant constraints can be eliminated according to the transitivity of partial order. 
To accomplish this, we compute a transitivity reduction~\cite{AhGa72} of a directed graph representing the poset. 
We demonstrate the effectiveness of our method through experiments using real-world clickstream data. 

The main contributions of this paper are as follows:
\begin{itemize}
\item We propose a shape-restricted optimization model for estimating item-choice probabilities from a user's previous PV sequence. 
This PV sequence model exploits the monotonicity constraints to precisely estimate item-choice probabilities. 
\item We derive two types of PV sequence posets according to the recency and frequency of a user's previous PVs. 
Experimental results show that the monotonicity constraints based on these posets greatly enhances the prediction performance of our PV sequence model. 
\item We devise constructive algorithms for transitive reduction specific to these posets. 
The time complexity of our algorithms is much smaller than that of general-purpose algorithms. 
Experimental results reveal that transitive reduction improves efficiency in terms of both the computation time and memory usage of our PV sequence model. 
\item We verify experimentally that higher prediction performance is achieved with our method than with the two-dimensional probability table and common machine learning methods, namely, logistic regression, artificial neural networks, and random forests. 
\end{itemize}

The remainder of this paper is organized as follows. 
Section~\ref{sec:2} gives a brief review of related work. 
Section~\ref{sec:3} describes the two-dimensional probability table~\cite{IwNi16}, and Section~4 presents our PV sequence model. 
Section~\ref{sec:4} describes our constructive algorithms for transitive reduction. 
Section~\ref{sec:5} evaluates the effectiveness of our method based on experimental results. 
Section~\ref{sec:6} concludes with a brief summary of this work and a discussion of future research directions.

\section{Related work}
\label{sec:2}

This section briefly surveys methods for predicting online user behavior and discusses some related work on shape-restricted regression. 

\subsection{Prediction of online user behavior}
A number of prior studies have aimed at predicting users' purchase behavior on e-commerce websites~\cite{CiHo19}. 
Mainstream research has applied stochastic or statistical models for predicting purchase sessions~\cite{BaHa18,KoLe20,MoFa04,MoSc04,PaPa16,SiBu04,VaBu05}, but these approaches do not consider which items users choose. 

Various machine learning methods have been used to predict online item-choice behavior, including logistic regression~\cite{DoJi19,ZhPe13}, association rule mining~\cite{PiZa10}, support vector machines~\cite{QiLi15,ZhPe13}, ensemble learning methods~\cite{LiGu15,LiZh15,RoSo15,YiWa15,ZhYa16}, and artificial neural networks~\cite{JeLu17,Vi15,WuTa15}. 
Tailored statistical models have also been proposed. For instance, Moe~\cite{Mo06} devised a two-stage multinomial logit model that separates the decision-making process into item views and purchase decisions.
Yao \textit{et al.}~\cite{YaKi17} proposed a joint framework consisting of user-level factor estimation and item-level factor aggregation based on the buyer decision process. 
Borges and Levener~\cite{BoLe07} used Markov chain models to estimate the probability of the next link choice of a user. 

These prior studies effectively utilized clickstream data in various prediction methods and
showed that consideration of time-evolving user behavior is crucial for precise prediction of online item-choice behavior. 
We therefore focused on user PV sequences to estimate item-choice probabilities on e-commerce websites. 
Moreover, we evaluated the prediction performance of our method by comparison with machine learning methods that have commonly been used in prior studies. 

\subsection{Shape-restricted regression}
In many practical situations, prior information is known about the relationship between explanatory and response variables. 
For instance, utility functions can be assumed to be increasing and concave according to economic theory~\cite{Ma91}, and option pricing functions to be monotone and convex according to finance theory~\cite{AiDu03}. 
Shape-restricted regression fits a nonparametric function to a set of given observations under shape restrictions such as monotonicity, convexity, concavity, or unimodality~\cite{ChGu15,GrJo14,GuSe18,WaGh12}. 

Isotonic regression is the most common method for shape-restricted regression. 
In general, isotonic regression is the problem of estimating a real-valued monotone (non-decreasing or non-increasing) function with respect to a given partial order of observations~\cite{PaXu99}. 
Some regularization techniques~\cite{GaKi18,TiHo11} and estimation algorithms~\cite{HaWa19,PaXu99,St15} have been proposed for isotonic regression. 

One of the greatest advantages of shape-restricted regression is that it mitigates overfitting, thereby improving prediction performance of regression models~\cite{AlRe05}. 
To utilize this advantage, Iwanaga \textit{et al.}~\cite{IwNi16} devised a shape-restricted optimization model for estimating item-choice probabilities on e-commerce websites. 
Along similar lines, we propose a shape-restricted optimization model based on order relations of PV sequences to improve prediction performance. 

\section{Two-dimensional probability table}
\label{sec:3}

This section briefly reviews the two-dimensional probability table proposed by Iwanaga \textit{et al.}~\cite{IwNi16}.

\subsection{Empirical probability table}

\begin{table*}[t]
\renewcommand{\arraystretch}{1.3}
\caption{Pageview History of Six User--Item Pairs}
\label{tab:PV}
\centering
\begin{tabular}{cccccccc} \toprule
      &     & \multicolumn{3}{c}{\#PVs} & Choice &
\\ \cmidrule(lr){3-5} \cmidrule(lr){6-6}
User & Item 	&Apr 1 &Apr 2 &Apr 3 &Apr 4 & 
$(r,f)$
&
$(v_1,v_2,v_3)$
\\ \midrule
$u_1$ & $i_2$ & 1      & 0      & 1      & 0      & $(3,2)$ & $(1,0,1)$       \\ 
$u_1$ & $i_4$ & 0      & 1      & 0      & 1      & $(2,1)$ & $(0,1,0)$       \\ 
$u_2$ & $i_1$ & 3      & 0      & 0      & 0      & $(1,3)$ & $(0,0,3)$       \\ 
$u_2$ & $i_3$ & 0      & 0      & 3      & 1      & $(3,3)$ & $(3,0,0)$       \\
$u_2$ & $i_4$ & 1      & 1      & 1      & 0      & $(3,3)$ & $(1,1,1)$       \\ 
$u_3$ & $i_2$ & 2      & 0      & 1      & 0      & $(3,3)$ & $(1,0,2)$       \\ \bottomrule
\end{tabular}
\end{table*}

Table~\ref{tab:PV} gives an example of a PV history for six user--item pairs. 
For instance, user $u_1$ viewed the page for item $i_2$ once on each of April 1 and 3. 
We focus on user choices (e.g., revisit and purchase) on April 4, which we call the \emph{base date}.
For instance, user $u_1$ chose item $i_4$ rather than item $i_2$ on the base date. 
We suppose for each user--item pair that recency and frequency are characterized by the last PV day and the total number of PVs, respectively. 
As shown in the table, the PV history can be summarized by the recency--frequency combination $(r,f) \in R \times F$, where $R$ and $F$ are index sets representing recency and frequency, respectively. 

Let $n_{rf}$ be the number of user--item pairs having $(r,f) \in R \times F$, and
set $q_{rf}$ to the number of choices occurring by user--item pairs that have $(r,f) \in R \times F$ on the base date.
In the case of Table~\ref{tab:PV}, the \emph{empirical probability table} is calculated as 
\begin{align} 
\left(\hat{x}_{rf} := \frac{q_{rf}}{n_{rf}}\right)_{\!\!(r,f) \in R \times F} 
&= 
\left(
\begin{array}{ccc}
0/0 & 0/0 & 0/1 \\
1/1 & 0/0 & 0/0 \\
0/0 & 0/1 & 1/3 \\
\end{array}
\right) \nonumber \\ 
&\approx 
\left(
\begin{array}{lll}
0.00    & 0.00 & 0.00    \\
1.00 & 0.00 & 0.00    \\
0.00    & 0.00 & 0.33 \\
\end{array}
\right), \label{eq:emp1}
\end{align}
where, for reasons of expediency, $\hat{x}_{rf} := 0$ for $(r,f) \in R \times F$ with $n_{rf} = 0$.

\subsection{Two-dimensional monotonicity model}
It is reasonable to assume that the recency and frequency of user--item pairs are positively associated with user item-choice probabilities. 
To estimate user item-choice probabilities $x_{rf}$ for all recency--frequency combinations $(r,f) \in R \times F$, the \emph{two-dimensional monotonicity model}~\cite{IwNi16} minimizes the weighted sum of squared errors under monotonicity constraints with respect to recency and frequency. 
\begin{align}
&\mathop{\mbox{minimize}}_{(x_{rf})_{(r,f) \in R \times F}} &&\sum_{(r,f) \in R \times F} n_{rf} (x_{rf} - \hat{x}_{rf})^2\label{obj:Mono}\\
&\mbox{~~~subject~to} 										&&x_{rf} \le x_{r+1,f}~~~((r,f) \in R \times F), \label{con1:Mono}\\
&&& x_{rf} \le x_{r,f+1}~~~((r,f) \in R \times F), \label{con2:Mono}\\
&&&0 \le x_{rf} \le 1~~~~\,((r,f) \in R \times F). \label{con3:Mono}
\end{align}

Note, however, that PV histories are often indistinguishable according to recency and frequency. 
A typical example is the set of user--item pairs $(u_2,i_3)$, $(u_2,i_4)$, and $(u_3,i_2)$ in Table~\ref{tab:PV}; although their PV histories are actually different, they have the same value $(r,f) = (3,3)$ for recency--frequency combinations. 
As described in the next section, we exploit the PV sequence to distinguish between such PV histories. 

\section{PV sequence model}
\label{sec:4}

This section presents our shape-restricted optimization model for estimating item-choice probabilities from a user's previous PV sequence. 

\subsection{PV sequence}

The \emph{PV sequence} for each user--item pair represents a time series of the number of PVs, and is written as  
\[
\bm{v} := (v_1, v_2, \ldots, v_n),
\]
where $v_j$ is the number of PVs $j$ periods earlier for $j=1,2,\ldots,n$ (see Table~\ref{tab:PV}). 
Note that sequence terms are arranged in reverse chronological order, so $v_j$ moves into the past as index $j$ increases. 

Throughout the paper, we express sets of consecutive integers as
\[
[m_1,m_2] := \{m_1,m_1+1,\ldots,m_2\} \subseteq \mathbb{Z}, 
\]
where $[m_1,m_2] = \emptyset$ when $m_1 > m_2$. 
The set of possible PV sequences is defined as 
\[
\Gamma := [0,m]^n = \{0,1,\ldots,m\}^n,
\]
where $m$ is the maximum number of PVs in each period, and $n$ is the number of periods considered. 

Our objective is to estimate item-choice probabilities $x_{\bm{v}}$ for all PV sequences $\bm{v} \in \Gamma$. 
However, the huge number of PV sequences makes it extremely difficult to accurately estimate such probabilities. 
In the case of $(n,m)=(|R|,|F|)=(5,6)$ for instance, the number of different PV sequences is $(m+1)^n = 16{,}807$, whereas the number of recency--frequency combinations is only $|R| \cdot |F| = 30$. 
To avoid this difficulty, we effectively utilize monotonicity constraints on item-choice probabilities as in the optimization model~\eqref{obj:Mono}--\eqref{con3:Mono}. 
In the next section, we introduce three operations underlying the development of the monotonicity constraints. 

\subsection{Operations based on recency and frequency}

From the perspective of frequency, it is reasonable to assume that  item-choice probability increases as the number of PVs in a particular period increases. 
To formulate this, we define the following operation: 
\begin{dfn}[$\texttt{Up}$]\label{def:Up}
On the domain 
\[
\mathcal{D}_{\normalfont \texttt{U}} :=\{(\bm{v},s) \in \Gamma \times [1,n] \mid v_s \le m-1 \}, 
\]
the function ${\normalfont \texttt{Up}}:\mathcal{D}_{\normalfont \texttt{U}} \to \Gamma$ is defined as
\[
((\ldots,v_s,\ldots),s) \mapsto (\ldots,v_s + 1,\ldots). 
\]
\end{dfn}
For instance, we have $\texttt{Up}((0,1,1),1) = (1,1,1)$, and $\texttt{Up}((1,1,1),2) = (1,2,1)$. 
Since this operation increases PV frequencies, the monotonicity constraint $x_{(0,1,1)} \le x_{(1,1,1)} \le x_{(1,2,1)}$ should be satisfied by item-choice probabilities. 

From the perspective of recency, we assume that more-recent PVs have a larger effect on increasing item-choice probability.
To formulate this, we consider the following operation for moving one PV from an old period to a new period: 
\begin{dfn}[$\texttt{Move}$] \label{def:Move}
On the domain
\[
\begin{array}{l}
\mathcal{D}_{\normalfont \texttt{M}} :=\{(\bm{v},s,t) \in \Gamma \times [1,n] \times [1,n] \\
\qquad\qquad\qquad\qquad\mid v_s \le m-1,~v_t \ge 1,~s < t \},
\end{array}
\] 
the function ${\normalfont \texttt{Move}}:\mathcal{D}_{\normalfont \texttt{M}} \to \Gamma$ is defined as
\[
((\ldots,v_s,\ldots,v_t,\ldots),s,t) \mapsto (\ldots,v_s + 1,\ldots,v_t - 1,\ldots). 
\]
\end{dfn}
For instance, we have $\texttt{Move}((1,1,1),2,3) = (1,2,0)$, and $\texttt{Move}((1,2,0),1,2) = (2,1,0)$. 
Because this operation increases the number of recent PVs, item-choice probabilities should satisfy the monotonicity constraint $x_{(1,1,1)} \le x_{(1,2,0)} \le x_{(2,1,0)}$. 

The PV sequence $\bm{v}=(1,1,1)$ represents a user's continued interest in a certain item over three periods. 
In contrast, the PV sequence $\bm{v}=(1,2,0)$ implies that a user's interest decreased over the two most-recent periods.
In this sense, the monotonicity constraint $x_{(1,1,1)} \le x_{(1,2,0)}$ may not be validated. 
Accordingly, we define the following alternative operation, which exchanges numbers of PVs to increase the number of recent PVs: 
\begin{dfn}[$\texttt{Swap}$] \label{def:Swap}
On the domain
\[
\mathcal{D}_{\normalfont \texttt{S}} :=\{(\bm{v},s,t) \in \Gamma \times [1,n] \times [1,n] \mid v_s <v_t,~s < t \},
\]
the function ${\normalfont \texttt{Swap}}:\mathcal{D}_{\normalfont \texttt{S}} \to \Gamma$ is defined as
\[
((\ldots,v_s,\ldots,v_t,\ldots),s,t) \mapsto (\ldots,v_t,\ldots,v_s,\ldots). 
\]
\end{dfn}
We thus have $\texttt{Swap}((1,0,2),2,3) = (1,2,0)$ because $v_2 < v_3$, and $\texttt{Swap}((1,2,0),1,2) = (2,1,0)$ because $v_1 < v_2$. 
Since this operation increases the number of recent PVs, item-choice probabilities should satisfy the monotonicity constraint $x_{(1,0,2)} \le x_{(1,2,0)} \le x_{(2,1,0)}$. 
Note that the monotonicity constraint $x_{(1,1,1)} \le x_{(1,2,0)}$ is not implied by this operation. 

\subsection{Partially ordered sets}

Let $U \subseteq \Gamma$ be a subset of PV sequences.
The image of each operation is then defined as  
\begin{align*}
\texttt{Up}(U) 	&= \{ \texttt{Up}(\bm{u},s) 		\mid \bm{u} \in U,~(\bm{u},s) \in \mathcal{D}_{\texttt{U}} \}, \notag \\
\texttt{Move}(U)	&= \{ \texttt{Move}(\bm{u},s,t) 	\mid \bm{u} \in U,~(\bm{u},s,t) \in \mathcal{D}_{\texttt{M}} \}, \notag \\
\texttt{Swap}(U)	&= \{ \texttt{Swap}(\bm{u},s,t) 	\mid \bm{u} \in U,~(\bm{u},s,t) \in \mathcal{D}_{\texttt{S}} \}. \notag
\end{align*}

We define $\texttt{UM}(U) := \texttt{Up}(U) \cup \texttt{Move}(U)$ for $U \subseteq \Gamma$. 
The following definition states that the binary relation $\bm{u} \prec_{\texttt{UM}} \bm{v}$ holds when $\bm{u}$ can be transformed into $\bm{v}$ by repeated application of \texttt{Up} and \texttt{Move}: 
\begin{dfn}[$\preceq_{\texttt{UM}}$] \label{def:UM}
Suppose $\bm{u},\bm{v} \in \Gamma$.
We write $\bm{u} \prec_{\normalfont \texttt{UM}} \bm{v}$ if and only if there exists $k \ge 1$ such that 
\[
\bm{v} \in {\normalfont \texttt{UM}}^k (\{\bm{u}\}) = \underbrace{{\normalfont \texttt{UM}} \circ \cdots \circ {\normalfont \texttt{UM}} \circ {\normalfont \texttt{UM}}}_{k~\mathrm{compositions}}(\{\bm{u}\}). 
\]
We also write $\bm{u} \preceq_{\normalfont \texttt{UM}} \bm{v}$ if $\bm{u} \prec_{\normalfont \texttt{UM}} \bm{v}$ or $\bm{u} = \bm{v}$.  
\end{dfn}

Similarly, we define $\texttt{US}(U) := \texttt{Up}(U) \cup \texttt{Swap}(U)$ for $U \subseteq \Gamma$. 
Then, the binary relation $\bm{u} \prec_{\texttt{US}} \bm{v}$ holds when $\bm{u}$ can be transformed into $\bm{v}$ by repeated application of \texttt{Up} and \texttt{Swap}. 
\begin{dfn}[$\preceq_{\texttt{US}}$] \label{def:US}
Suppose $\bm{u},\bm{v} \in \Gamma$.
We write $\bm{u} \prec_{\normalfont \texttt{US}} \bm{v}$ if and only if there exists $k \ge 1$ such that 
\[
\bm{v} \in {\normalfont \texttt{US}}^k (\{\bm{u}\}) = \underbrace{{\normalfont \texttt{US}} \circ \cdots \circ {\normalfont \texttt{US}} \circ {\normalfont \texttt{US}}}_{k~\mathrm{compositions}}(\{\bm{u}\}). 
\]
We also write $\bm{u} \preceq_{\normalfont \texttt{US}} \bm{v}$ if $\bm{u} \prec_{\normalfont \texttt{US}} \bm{v}$ or $\bm{u} = \bm{v}$.   
\end{dfn}

To prove properties of these binary relations, we can use the lexicographic order, which is a well-known linear order~\cite{Sc16}: 
\begin{dfn}[$\preceq_{\texttt{lex}}$] \label{def:lex}
Suppose $\bm{u},\bm{v} \in \Gamma$.
We write $\bm{u} \prec_{\normalfont \texttt{lex}} \bm{v}$ if and only if there exists $s \in [1,n]$ such that $u_s < v_s$ and $u_j = v_j$ for $j \in [1,s-1]$. 
We also write $\bm{u} \preceq_{\normalfont \texttt{lex}} \bm{v}$ if $\bm{u} \prec_{\normalfont \texttt{lex}} \bm{v}$ or $\bm{u} = \bm{v}$. 
\end{dfn}

Each application of \texttt{Up}, \texttt{Move}, and \texttt{Swap} makes a PV sequence greater in the lexicographic order. 
Therefore, we can obtain the following lemma: 
\begin{lem} \label{lem:lex}
Suppose $\bm{u},\bm{v} \in \Gamma$.  
If $\bm{u} \preceq_{\normalfont \texttt{UM}} \bm{v}$ or $\bm{u} \preceq_{\normalfont \texttt{US}} \bm{v}$, then $\bm{u} \preceq_{\normalfont \texttt{lex}} \bm{v}$. 
\end{lem}

The following theorem states that a partial order of PV sequences is derived by operations \texttt{Up} and \texttt{Move}. 
\begin{thm} \label{thm:posetUM}
The pair $(\Gamma,\preceq_{\normalfont \texttt{UM}})$ is a poset. 
\end{thm}
\begin{proof}
From Definition~\ref{def:UM}, the relation $\preceq_{\texttt{UM}}$ is reflexive and transitive. 
Suppose $\bm{u} \preceq_{\texttt{UM}} \bm{v}$ and $\bm{v} \preceq_{\texttt{UM}} \bm{u}$.
It follows from Lemma~\ref{lem:lex} that $\bm{u} \preceq_{\texttt{lex}} \bm{v}$ and $\bm{v} \preceq_{\texttt{lex}} \bm{u}$. 
Since the relation $\preceq_{\texttt{lex}}$ is antisymmetric, we have $\bm{u}=\bm{v}$, thus proving that the relation $\preceq_{\texttt{UM}}$ is also antisymmetric. 
\end{proof}

We can similarly prove the following theorem for operations \texttt{Up} and \texttt{Swap}:
\begin{thm} \label{thm:US}
The pair $(\Gamma,\preceq_{\normalfont \texttt{US}})$ is a poset. 
\end{thm}

\subsection{Shape-restricted optimization model}

Let $n_{\bm{v}}$ be the number of user--item pairs that have the PV sequence $\bm{v} \in \Gamma$.
Also, $q_{\bm{v}}$ is the number of choices arising from user--item pairs having $\bm{v} \in \Gamma$ on the base date. 
Similarly to Eq.~\eqref{eq:emp1}, we can calculate empirical item-choice probabilities as 
\begin{align} \label{eq:emp2}
\hat{x}_{\bm{v}} := \frac{q_{\bm{v}}}{n_{\bm{v}}} \qquad (\bm{v} \in \Gamma). 
\end{align}

Our shape-restricted optimization model minimizes the weighted sum of squared errors subject to the monotonicity constraint: 
\begin{align}
\mathop{\mbox{minimize}}_{(x_{\bm{v}})_{\bm{v} \in \Gamma}} 
& \quad \sum_{\bm{v} \in \Gamma} n_{\bm{v}} (x_{\bm{v}} - \hat{x}_{\bm{v}})^2 \label{obj:PVS} \\
\mbox{subject~to} 
& \quad x_{\bm{u}} \le x_{\bm{v}} \qquad (\bm{u},\bm{v} \in \Gamma \mbox{~with~} \bm{u} \prec \bm{v}), \label{con1:PVS} \\
& \quad 0 \le x_{\bm{v}} \le 1 ~~~(\bm{v} \in \Gamma), \label{con2:PVS}
\end{align}
where $\bm{u} \prec \bm{v}$ in Eq.~\eqref{con1:PVS} is defined by one of the partial orders $\prec_{\texttt{UM}}$ or $\prec_{\texttt{US}}$. 

The monotonicity constraint~\eqref{con1:PVS} enhances the estimation accuracy of item-choice probabilities. 
In addition, our shape-restricted optimization model can be used in a post-processing step to improve prediction performance of other machine learning methods. 
Specifically, we first compute item-choice probabilities using a machine learning method and 
then substitute the computed values into $(\hat{x}_{\bm{v}})_{\bm{v} \in \Gamma}$ to solve the optimization model~\eqref{obj:PVS}--\eqref{con2:PVS}. 
Consequently, we can obtain item-choice probabilities corrected by the monotonicity constraint~\eqref{con1:PVS}. 
Section~\ref{sec:6.4} illustrates the usefulness of this approach. 

However, since $|\Gamma| = (m+1)^n$, it follows that the number of constraints in Eq.~\eqref{con1:PVS} is $\mathcal{O}((m+1)^{2n})$, which can be extremely large. 
When $(n,m) = (5,6)$, for instance, we have $(m+1)^{2n} = 282{,}475{,}249$. 
The next section describes how we mitigate this difficulty by removing redundant constraints in Eq.~\eqref{con1:PVS}.

\begin{figure*}[h]
\centering
\begin{tabular}{c}
\includegraphics[keepaspectratio, scale=0.42, bb=100 30 1100 700]{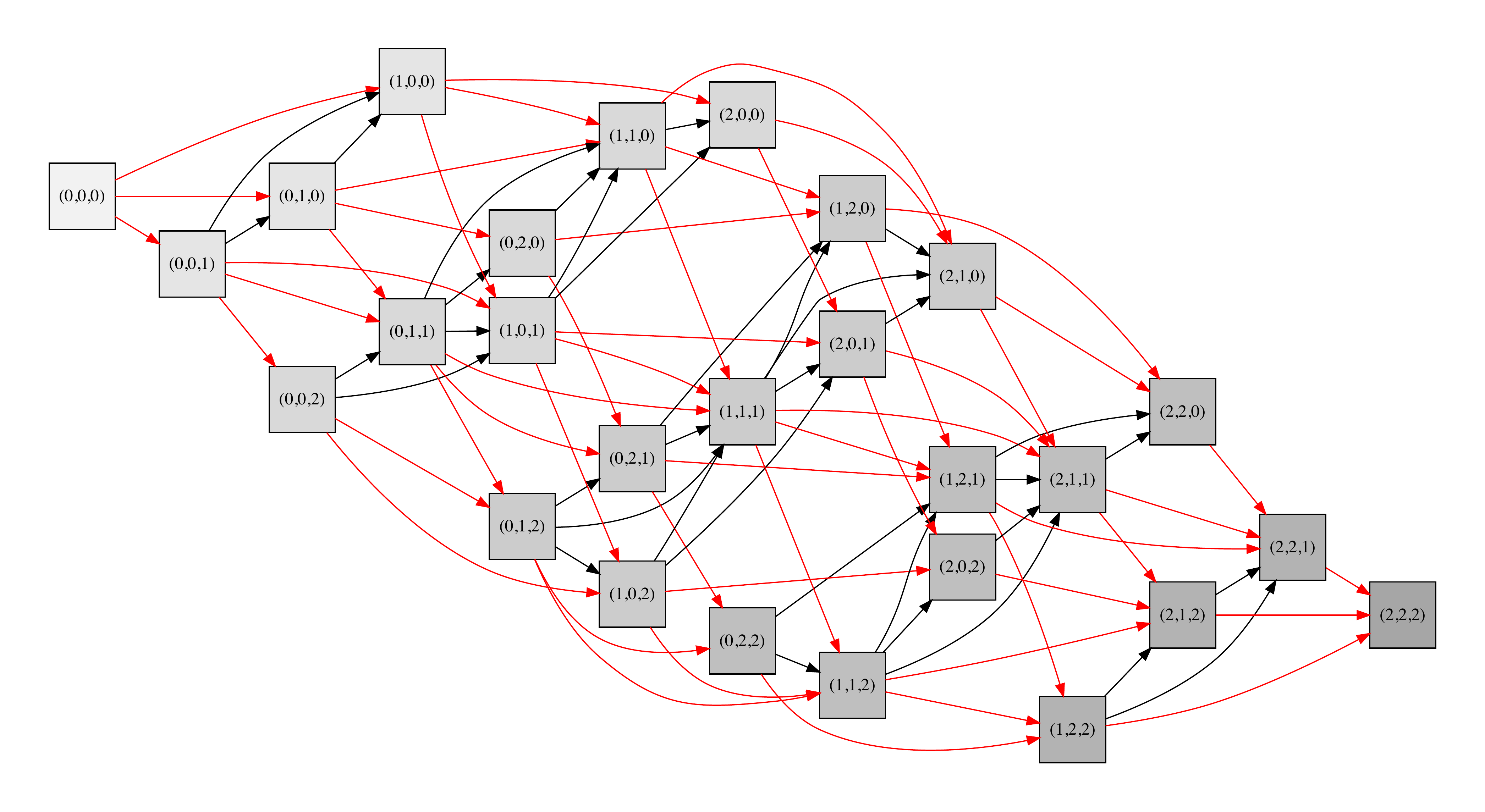}\\[.05in]
(a)~Operation-based graph\\[.2in]
\includegraphics[keepaspectratio, scale=0.42, bb=100 30 1100 250]{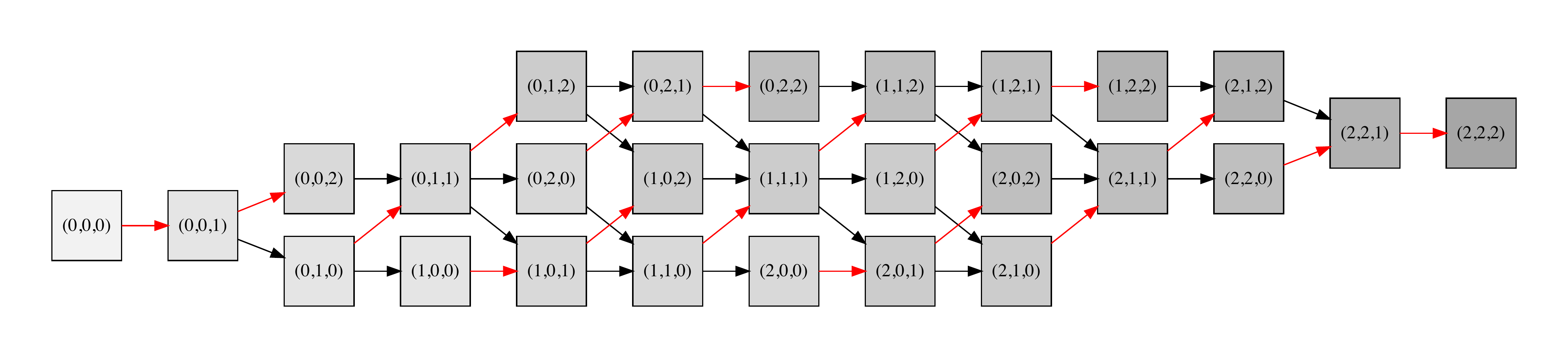}\\[.05in]
(b)~Transitive reduction
\end{tabular}
\caption{Directed graph representations of the poset $(\Gamma, \preceq_{\texttt{UM}})$ with $(n,m)=(3,2)$.}
\label{fig:hasse1_3}
\end{figure*}
\begin{figure*}[h]
\centering
\begin{tabular}{c}
\includegraphics[keepaspectratio, scale=0.42, bb=100 30 1100 700]{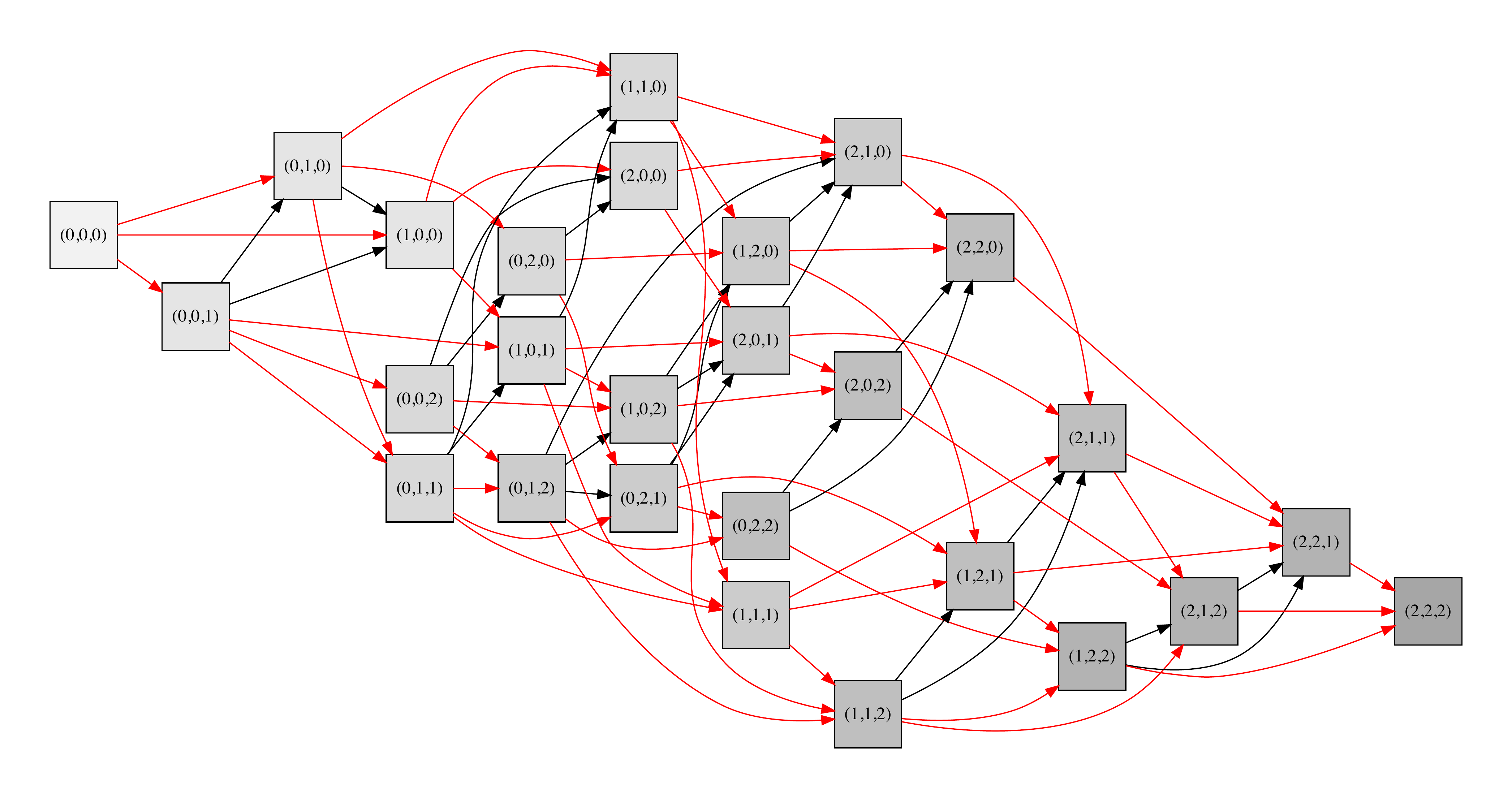}\\[.05in]
(a)~Operation-based graph\\[.2in]
\includegraphics[keepaspectratio, scale=0.42, bb=100 30 1100 450]{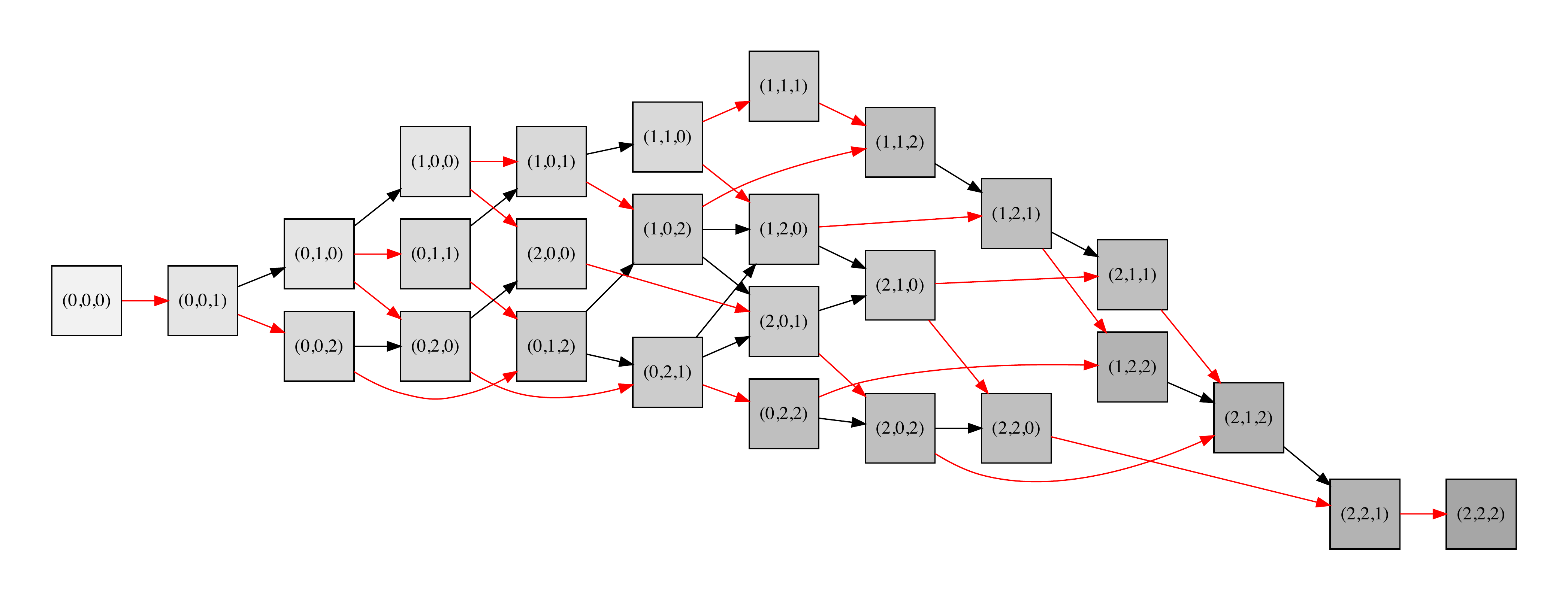}\\[.05in]
(b)~Transitive reduction
\end{tabular}
\caption{Directed graph representations of the poset $(\Gamma, \preceq_{\texttt{US}})$ with $(n,m)=(3,2)$.}\label{fig:hasse2_3}
\end{figure*}

\section{Algorithms for transitive reduction}
\label{sec:5}

This section describes our constructive algorithms for transitive reduction to decrease the problem size in our shape-restricted optimization model. 

\subsection{Transitive reduction}

A poset $(\Gamma,\preceq)$ can be represented by a directed graph $(\Gamma, E)$, where $\Gamma$ and $E \subseteq \Gamma \times \Gamma$ are  sets of nodes and directed edges, respectively. 
Each directed edge $(\bm{u},\bm{v}) \in E$ in this graph corresponds to the order relation $\bm{u} \prec \bm{v}$, so the number of directed edges coincides with the number of constraints in Eq.~\eqref{con1:PVS}. 

Figs.~\ref{fig:hasse1_3} and \ref{fig:hasse2_3} show directed graph representations of posets $(\Gamma, \preceq_{\texttt{UM}})$ and $(\Gamma, \preceq_{\texttt{US}})$, respectively.
Each edge in Figs.~\ref{fig:hasse1_3}(a) and \ref{fig:hasse2_3}(a) corresponds to one of the operations \texttt{Up}, \texttt{Move}, or \texttt{Swap}. Edge $(\bm{u},\bm{v})$ is red if $\bm{v} \in \texttt{Up}(\{\bm{u}\})$ and black if $\bm{v} \in \texttt{Move}(\{\bm{u}\})$ or $\bm{v} \in \texttt{Swap}(\{\bm{u}\})$. 
The directed graphs in Figs.~\ref{fig:hasse1_3}(a) and \ref{fig:hasse2_3}(a) can be easily created. 

Suppose there are three edges \[(\bm{u},\bm{w}),(\bm{w},\bm{v}),(\bm{u},\bm{v}) \in E.\]
In this case, edge $(\bm{u},\bm{v})$ is implied by the other edges due to the transitivity of partial order
\[
\langle \bm{u} \prec \bm{w},~\bm{w} \prec \bm{v} \rangle ~\Rightarrow~ \bm{u} \prec \bm{v}, 
\]
or, equivalently, 
\[
\langle x_{\bm{u}} \le x_{\bm{w}},~x_{\bm{w}} \le x_{\bm{v}} \rangle ~\Rightarrow~ x_{\bm{u}} \le x_{\bm{v}}.
\]
As a result, edge $(\bm{u},\bm{v})$ is redundant and can be removed from the directed graph.  

A \emph{transitive reduction}, also known as a Hasse diagram, of a directed graph $(\Gamma, E)$ is its subgraph $(\Gamma, E^*)$ such that all redundant edges are removed using the transitivity of partial order~\cite{AhGa72}. 
Figs.~\ref{fig:hasse1_3}(b) and \ref{fig:hasse2_3}(b) show transitive reductions of the directed graphs shown in Figs.~\ref{fig:hasse1_3}(a)~and~\ref{fig:hasse2_3}(a), respectively. 
By computing transitive reductions, the number of edges is reduced from 90 to 42 in Fig.~\ref{fig:hasse1_3}, and from 81 to 46 in Fig.~\ref{fig:hasse2_3}. 
This transitive reduction is known to be unique~\cite{AhGa72}. 

\subsection{General-purpose algorithms}

The transitive reduction $(\Gamma,E^*)$ is characterized by the following lemma~\cite{Sc16}: 
\begin{lem} \label{lem:rdc}
Suppose $(\bm{u},\bm{v}) \in \Gamma \times \Gamma$.
Then, $(\bm{u},\bm{v}) \in E^*$ holds if and only if both of the following conditions are fulfilled: 
\begin{itemize}
\item[] {\rm (C1)} $\bm{u} \prec \bm{v}$, and 
\item[] {\rm (C2)} if $\bm{w} \in \Gamma$ satisfies $\bm{u} \preceq \bm{w} \preceq \bm{v}$, then $\bm{w} \in \{\bm{u},\bm{v}\}$.
\end{itemize}
\end{lem}

The basic strategy in general-purpose algorithms for transitive reduction involves the following steps: 
\begin{itemize}
\item Step 1: An exhaustive directed graph $(\Gamma, E)$ is generated from a given poset $(\Gamma,\preceq)$. 
\item Step 2: The transitive reduction $(\Gamma,E^*)$ is computed from the directed graph $(\Gamma, E)$ using Lemma~\ref{lem:rdc}. 
\end{itemize}

Various algorithms for speeding up the computation in Step~2 have been proposed. 
Recall that $|\Gamma| = (m+1)^n$ in our situation. 
Warshall's algorithm~\cite{Wa62} has time complexity $\mathcal{O}((m+1)^{3n})$ for completing Step~2~\cite{Sc16}. 
By using a sophisticated algorithm for fast matrix multiplication, this time complexity can be reduced to $\mathcal{O}((m+1)^{2.3729n})$~\cite{Le14}. 

However, such general-purpose algorithms are clearly inefficient, especially when $n$ is very large, 
and Step~1 requires a huge number of computations. 
To resolve this difficulty, we devised specialized algorithms for directly constructing a transitive reduction. 

\subsection{Constructive algorithms}

Let $(\Gamma, E^*_{\texttt{UM}})$ be a transitive reduction of a directed graph $(\Gamma, E_{\texttt{UM}})$ representing the poset $(\Gamma, \preceq_{\texttt{UM}})$. 
Then, the transitive reduction can be characterized by the following theorem: 

\begin{thm} \label{thm:iffUM}
Suppose $(\bm{u},\bm{v}) \in \Gamma \times \Gamma$. 
Then, $(\bm{u},\bm{v}) \in E^*_{\normalfont \texttt{UM}}$ holds if and only if any one of the following conditions is fulfilled:   
\begin{itemize}
\item[] {\rm (UM1)} $\bm{v} = {\normalfont \texttt{Up}}(\bm{u}, n)$, or
\item[] {\rm (UM2)} 
$\exists s \in [1,n]$ such that $\bm{v} = {\normalfont \texttt{Move}}(\bm{u},s,s+1)$. 
\end{itemize}
\end{thm}
\begin{proof}
See Appendix~\ref{app:a1}. 
\end{proof}

Theorem~\ref{thm:iffUM} gives a constructive algorithm that directly computes the transitive reduction $(\Gamma, E^*_{\texttt{UM}})$ without generating an exhaustive directed graph $(\Gamma, E)$. 
Our algorithm is based on the breadth-first search~\cite{CoLe09}. 
Specifically, we start with a node list $L=\{(0,0,\ldots,0)\} \subseteq \Gamma$. 
At each iteration of the algorithm, we choose $\bm{u} \in L$, enumerate $\bm{v} \in \Gamma$ such that $(\bm{u},\bm{v}) \in E^*_{\texttt{UM}}$, and add these nodes to $L$. 

Table~\ref{tab:TR-UM} shows this enumeration process for $\bm{u} = (0,2,1)$ with $(n,m)=(3,2)$. 
The operations \texttt{Up} and \texttt{Move} generate \[\bm{v} \in \{(1,2,1),(0,2,2),(1,1,1),(1,2,0)\},\] which amounts to searching edges $(\bm{u},\bm{v})$ in Fig.~\ref{fig:hasse1_3}(a). 
We next check whether each $\bm{v}$ satisfies conditions (UM1) or (UM2) in Theorem~\ref{thm:iffUM}. 
As shown in Table~\ref{tab:TR-UM}, we choose \[\bm{v} \in \{(0,2,2),(1,1,1)\}\] and add them to list $L$; this amounts to enumerating edges $(\bm{u},\bm{v})$ in Fig.~\ref{fig:hasse1_3}(b). 

\begin{table}[t]
\renewcommand{\arraystretch}{1.3}
\centering
\caption{Process of enumerating $\bm{v} \in \Gamma$ such that $(\bm{u},\bm{v}) \in E^*_{\texttt{UM}}$}
\label{tab:TR-UM}
\begin{tabular}{ccccc} \toprule
$\bm{u}$  & Operation                   & $\bm{v}$  & (UM1)    & (UM2)    \\ \midrule
$(0,2,1)$ & $\texttt{Up}(\bm{u},1)$     & $(1,2,1)$ & $\times$ & ---      \\ 
          & $\texttt{Up}(\bm{u},3)$     & $(0,2,2)$ & $\surd$  & ---      \\
          & $\texttt{Move}(\bm{u},1,2)$ & $(1,1,1)$ & ---      & $\surd$  \\
          & $\texttt{Move}(\bm{u},1,3)$ & $(1,2,0)$ & ---      & $\times$ \\ \bottomrule
\end{tabular}
\end{table}

Appendix~\ref{app1um} presents pseudocode for our constructive algorithm (Algorithm~\ref{app:b1}).
Recalling the time complexity analysis of the breadth-first search~\cite{CoLe09}, one readily sees that the time complexity of Algorithm~\ref{app:b1} is $\mathcal{O}(n(m+1)^n)$, which is much smaller than $\mathcal{O}((m+1)^{2.3729n})$ as achieved by the general-purpose algorithm~\cite{Le14}, especially when $n$ is very large. 


Next, we focus on the transitive reduction $(\Gamma, E^*_{\texttt{US}})$ of a directed graph $(\Gamma, E_{\texttt{US}})$ representing the poset $(\Gamma, \preceq_{\texttt{US}})$. 
The transitive reduction can then be characterized by the following theorem: 
\begin{thm} \label{thm:iffUS}
Suppose $(\bm{u},\bm{v}) \in \Gamma \times \Gamma$.
Then, $(\bm{u},\bm{v}) \in E^*_{\texttt{US}}$ holds if and only if any one of the following conditions is fulfilled:   
\begin{itemize}
\item[] {\rm (US1)} $\exists s \in [1,n]$ such that $\bm{v} = {\normalfont \texttt{Up}}(\bm{u},s)$ and $u_j \not\in \{u_s, u_s + 1\}$ for all $j \in [s+1,n]$, or
\item[] {\rm (US2)} $\exists (s,t) \in [1,n] \times [1,n]$ such that $\bm{v} = {\normalfont \texttt{Swap}}(\bm{u},s,t)$ and $u_j \not\in [u_s, u_t]$ for all $j \in [s+1,t-1]$. 
\end{itemize}
\end{thm}
\begin{proof}
See Appendix~\ref{app:a2}. 
\end{proof}

Theorem~\ref{thm:iffUS} also gives a constructive algorithm for computing the transitive reduction $(\Gamma, E^*_{\texttt{US}})$. 
Let us again consider $\bm{u} = (0,2,1)$ as an example with $(n,m)=(3,2)$. 
As shown in Table~\ref{tab:TR-US}, operations \texttt{Up} and \texttt{Swap} generate \[\bm{v} \in \{(1,2,1),(0,2,2),(2,0,1),(1,2,0)\},\] and we choose \[\bm{v} \in \{(0,2,2),(2,0,1),(1,2,0)\}\]
(see also Figs.~\ref{fig:hasse2_3}(a)~and~\ref{fig:hasse2_3}(b)).

\begin{table}[t]
\renewcommand{\arraystretch}{1.3}
\centering
\caption{Process of enumerating $\bm{v} \in \Gamma$ such that $(\bm{u},\bm{v}) \in E^*_{\texttt{US}}$}
\label{tab:TR-US}
\begin{tabular}{ccccc} \toprule
$\bm{u}$  & Operation                   & $\bm{v}$  & (US1)    & (US2)   \\ \midrule
$(0,2,1)$ & $\texttt{Up}(\bm{u},1)$     & $(1,2,1)$ & $\times$ & ---     \\ 
          & $\texttt{Up}(\bm{u},3)$     & $(0,2,2)$ & $\surd$  & ---     \\
          & $\texttt{Swap}(\bm{u},1,2)$ & $(2,0,1)$ & ---      & $\surd$ \\
          & $\texttt{Swap}(\bm{u},1,3)$ & $(1,2,0)$ & ---      & $\surd$ \\ \bottomrule
\end{tabular}
\end{table}

Appendix~\ref{app1us} presents pseudocode 
for our constructive algorithm (Algorithm~\ref{app:b2}). 
Its time complexity is estimated to be $\mathcal{O}(n^2(m+1)^n)$, which is larger than that of Algorithm~\ref{app:b1} but much smaller than that of the general-purpose algorithm~\cite{Le14}, especially when $n$ is very large. 

\section{Experiments}
\label{sec:6}
The experimental results reported in this section evaluate the effectiveness of our method for estimating item-choice probabilities. 

We used real-world clickstream data collected from a Chinese e-commerce website, Tmall\footnote{\url{https://tianchi.aliyun.com/dataset/}}. 
We used a dataset\footnote{\url{https://www.dropbox.com/sh/dbzmtq4zhzbj5o9/AACldzQWbw-igKjcPTBI6ZPAa?dl=0}} preprocessed by Ludewig and Jannach~\cite{LuJa18}.  
Each record corresponds to one PV and contains information such as user ID, item ID, and a timestamp. 
The dataset includes 28,316,459 unique user--item pairs composed from 422,282 users and 624,221 items. 

\subsection{Methods for comparison}
\label{sec:6.1}
\begin{table}[t]
\renewcommand{\arraystretch}{1.3}
\centering
\caption{Methods for comparison}
\label{tbl:methods}
\footnotesize
\begin{tabular}{ll} \toprule
Abbr.  & Method \\ \midrule
2dimEmp  & Empirical probability table~\eqref{eq:emp1}~\cite{IwNi16} \\
2dimMono & Two-dimensional monotonicity model~\eqref{obj:Mono}--\eqref{con3:Mono}~\cite{IwNi16} \\
SeqEmp   & Empirical probabilities~\eqref{eq:emp2} for PV sequences \\
SeqUM    & Our PV sequence model~\eqref{obj:PVS}--\eqref{con2:PVS} using $(\Gamma,\preceq_{\texttt{UM}})$ \\
SeqUS    & Our PV sequence model~\eqref{obj:PVS}--\eqref{con2:PVS} using $(\Gamma,\preceq_{\texttt{US}})$ \\
LR       & $L_2$-regularized logistic regression \\
ANN      & Artificial neural networks for regression \\
RF       & Random forest of regression trees \\
\bottomrule
\end{tabular}
\end{table}

We compared the performance of the methods listed in Table~\ref{tbl:methods}.
All computations were performed on an Apple MacBook Pro computer with an Intel Core i7-5557U CPU (3.10~GHz) and 16~GB of memory.

The optimization models~\eqref{obj:Mono}--\eqref{con3:Mono} and \eqref{obj:PVS}--\eqref{con2:PVS} were solved using OSQP\footnote{\url{https://osqp.org/docs/index.html}}~\cite{StBaXX}, a numerical optimization package for solving convex quadratic optimization problems. 
As in Table~\ref{tab:PV}, daily-PV sequences were 
calculated for each user--item pair, where $m$ is the maximum number of daily PVs and $n$ is the number of terms (past days) in the PV sequence. 
In this process, all PVs from more than $n$ days earlier 
were 
added to the number of PVs $n$ days earlier, and numbers of daily PVs exceeding $m$ were rounded down to $m$. 
Similarly, the recency--frequency combinations $(r,f) \in R \times F$ were calculated using daily PVs as in Table~\ref{tab:PV}, where $(|R|,|F|) = (n,m)$. 

Other machine learning methods (LR, ANN, and RF) were respectively implemented using the \textsf{LogisticRegressionCV}, \textsf{MLPRegressor}, and \textsf{RandomForestRegressor} functions in scikit-learn, a Python library of machine learning tools. 
Related hyperparameters were tuned through 3-fold cross-validation according to the parameter settings in a benchmark study~\cite{OrCa18}. 
These machine learning methods employed the PV sequence $(v_1,v_2,\ldots,v_n)$ as $n$ input variables for computing item-choice probabilities. 
We standardized each input variable and performed undersampling to improve prediction performance. 

\subsection{Performance evaluation methodology}
\label{sec:6.2}

There are five pairs of training and validation sets of clickstream data in the preprocessed dataset~\cite{LuJa18}. 
As shown in Table~\ref{tbl:periods}, each training period is 90 days, and the next day is the validation period. 
The first four pairs of training and validation sets, which we call the \emph{training set}, were used for model estimation, and the fifth pair was used for performance evaluation. 
To examine how sample size affects prediction performance, we prepared small-sample training sets by randomly choosing user--item pairs from the original training set. 
Here, the sampling rates are 0.1\%, 1\%, and 10\%, and the original training set is referred to as the full sample. 
Note that the results were averaged over 10 trials for the sampled training sets. 

\begin{table}[t]
\renewcommand{\arraystretch}{1.3}
\centering
\caption{Training and validation periods}
\label{tbl:periods}
\footnotesize
\begin{tabular}{clll} \toprule
        & \multicolumn{2}{c}{Training}     &                   \\ \cmidrule(lr){2-3}
Pair ID & \multicolumn{1}{c}{Start} & \multicolumn{1}{c}{End} & \multicolumn{1}{c}{Validation} \\ \midrule
1       & 21 May 2015  & 18 August 2015    & 19 August 2015    \\ 
2       & 31 May 2015  & 28 August 2015    & 29 August 2015    \\ 
3       & 10 June 2015 &  7 September 2015 &  8 September 2015 \\ 
4       & 20 June 2015 & 17 September 2015 & 18 September 2015 \\ 
5       & 30 June 2015 & 27 September 2015 & 28 September 2015 \\ 
\bottomrule
\end{tabular}
\end{table}
~
\begin{table*}[t]
\renewcommand{\arraystretch}{1.3}
\centering
\caption{Problem size of our PV sequence model~\eqref{obj:PVS}--\eqref{con2:PVS}}
\label{tbl:problem_size}
\footnotesize
\begin{tabular}{ccrrrrrrr} \toprule
    &     &        & \multicolumn{6}{c}{\#Cons in Eq.~\eqref{con1:PVS}} \\ \cmidrule(lr){4-9}
    &     &        & \multicolumn{2}{c}{Enumeration} & \multicolumn{2}{c}{Operation} & \multicolumn{2}{c}{Reduction} \\ \cmidrule(lr){4-5}\cmidrule(lr){6-7}\cmidrule(lr){8-9}
$n$ & $m$ & \#Vars & SeqUM      & SeqUS      & SeqUM   & SeqUS   & SeqUM  & SeqUS  \\ \midrule
5   & 1   & 32     & 430        & 430        & 160     & 160     & 48     & 48     \\
5   & 2   & 243    & 21,383     & 17,945     & 1,890   & 1,620   & 594    & 634    \\
5   & 3   & 1,024  & 346,374    & 255,260    & 9,600   & 7,680   & 3,072  & 3,546  \\
5   & 4   & 3,125  & 3,045,422  & 2,038,236  & 32,500  & 25,000  & 10,500 & 12,898 \\
5   & 5   & 7,776  & 18,136,645 & 11,282,058 & 86,400  & 64,800  & 28,080 & 36,174 \\
5   & 6   & 16,807 & 82,390,140 & 48,407,475 & 195,510 & 144,060 & 63,798 & 85,272 \\ \midrule
1   & 6   & 7      & 21         & 21         & 6       & 6       & 6      & 6      \\
2   & 6   & 49     & 1,001      & 861        & 120     & 105     & 78     & 93     \\
3   & 6   & 343    & 42,903     & 32,067     & 1,638   & 1,323   & 798    & 1,018  \\
4   & 6   & 2,401  & 1,860,622  & 1,224,030  & 18,816  & 14,406  & 7,350  & 9,675  \\
5   & 6   & 16,807 & 82,390,140 & 48,407,475 & 195,510 & 144,060 & 63,798 & 85,272 \\ \bottomrule
\end{tabular}
\end{table*}

We considered the \emph{top-$N$ selection} task to evaluate prediction performance. 
Specifically, we focused on items that were viewed by a particular user during a training period.
From among these items, we selected $I_{\rm sel}$, a set of top-$N$ items for the user according to estimated item-choice probabilities. 
The most-recently viewed items were selected when two or more items had the same choice probability. 
Let $I_{\rm view}$ be the set of items viewed by the user in the validation period. 
Then, the \emph{F1 score} is defined by the harmonic average of $\mbox{\emph{Recall}} := |I_{\rm sel} \cap I_{\rm view}|/|I_{\rm view}|$ and $\mbox{\emph{Precision}} := |I_{\rm sel} \cap I_{\rm view}|/|I_{\rm sel}|$ as
\[
\mathrm{F1~score} := \frac{2 \cdot \mathrm{Recall} \cdot \mathrm{Precision}}{\mathrm{Recall} + \mathrm{Precision}}.
\]
In the following sections, we examine F1 scores averaged over all users. 
The percentage of user--item pairs leading to item choices is only 0.16\%.

\subsection{Effects of the transitive reduction}
\label{sec:6.3}

We generated constraints in Eq.~\eqref{con1:PVS} based on the following three directed graphs:  
\begin{itemize}
\item Case 1~(Enumeration): All edges $(\bm{u},\bm{v})$ satisfying $\bm{u} \prec \bm{v}$ were enumerated.  
\item Case 2~(Operation): Edges corresponding to operations \texttt{Up}, \texttt{Move}, and \texttt{Swap} were generated as in Figs.~\ref{fig:hasse1_3}(a) and \ref{fig:hasse2_3}(a). 
\item Case 3~(Reduction): Transitive reduction was computed using our algorithms as in Figs.~\ref{fig:hasse1_3}(b) and \ref{fig:hasse2_3}(b). 
\end{itemize}

\begin{table*}[h]
\renewcommand{\arraystretch}{1.3}
\centering
\caption{Computation times for our PV sequence model~\eqref{obj:PVS}--\eqref{con2:PVS}}
\label{tbl:comp_time}
\begin{tabular}{ccrrrrrrr} \toprule
    &     &        & \multicolumn{6}{c}{Time [s]} \\ \cmidrule(lr){4-9}
    &     &        & \multicolumn{2}{c}{Enumeration} & \multicolumn{2}{c}{Operation} & \multicolumn{2}{c}{Reduction} \\ \cmidrule(lr){4-5}\cmidrule(lr){6-7}\cmidrule(lr){8-9}
$n$ & $m$ & \#Vars & SeqUM  & SeqUS & SeqUM  & SeqUS  & SeqUM  & SeqUS  \\ \midrule
5   & 1   & 32     & 0.00   & 0.01  & 0.00   & 0.00   & 0.00   & 0.00   \\
5   & 2   & 243    & 2.32   & 1.66  & 0.09   & 0.07   & 0.03   & 0.02   \\
5   & 3   & 1,024  & 558.22 & 64.35 & 3.41   & 0.71   & 0.13   & 0.26   \\
5   & 4   & 3,125  & OM     & OM    & 24.07  & 13.86  & 1.72   & 5.80   \\
5   & 5   & 7,776  & OM     & OM    & 180.53 & 67.34  & 9.71   & 36.94  \\
5   & 6   & 16,807 & OM     & OM    & 906.76 & 522.84 & 86.02  & 286.30 \\ \midrule
1   & 6   & 7      & 0.00   & 0.00  & 0.00   & 0.00   & 0.00   & 0.00   \\
2   & 6   & 49     & 0.03   & 0.01  & 0.01   & 0.00   & 0.00   & 0.00   \\
3   & 6   & 343    & 12.80  & 1.68  & 0.20   & 0.03   & 0.05   & 0.02   \\
4   & 6   & 2,401  & OM     & OM    & 8.07   & 4.09   & 2.12   & 2.87   \\
5   & 6   & 16,807 & OM     & OM    & 906.76 & 522.84 & 86.02  & 286.30 \\ \bottomrule
\end{tabular}
\end{table*}
\begin{table*}[h]
\renewcommand{\arraystretch}{1.3}
\centering
\caption{Computational performance of our PV sequence model~\eqref{obj:PVS}--\eqref{con2:PVS}}
\label{tbl:perform}
\begin{tabular}{ccrrrrrrrr} \toprule
    &     &        & \multicolumn{2}{c}{\#Cons in Eq.~\eqref{con1:PVS}} & \multicolumn{2}{c}{Time [s]} & \multicolumn{3}{c}{F1 score [\%], $N=3$}    \\ \cmidrule(lr){4-5} \cmidrule(lr){6-7} \cmidrule(lr){8-10}
$n$ & $m$ & \#Vars & SeqUM  & SeqUS   & SeqUM  & SeqUS  & SeqEmp & SeqUM & SeqUS \\ \midrule
3   & 30  & 29,791 & 84,630 & 118,850 &  86.72 & 241.46 & 12.25  & 12.40 & 12.40 \\
4   & 12  & 28,561 & 99,372 & 142,800 & 198.82 & 539.76 & 12.68  & 12.93 & 12.95 \\
5   & 6   & 16,807 & 63,798 & 85,272  &  86.02 & 286.30 & 12.90  & 13.18 & 13.18 \\
6   & 4   & 15,625 & 62,500 & 76,506  &  62.92 & 209.67 & 13.14  & 13.49 & 13.48 \\
7   & 3   & 16,384 & 67,584 & 76,818  &  96.08 & 254.31 & 13.23  & 13.52 & 13.53 \\
8   & 2   &  6,561 & 24,786 & 25,879  &  19.35 &  17.22 & 13.11  & 13.37 & 13.35 \\
9   & 2   & 19,683 & 83,106 & 86,386  & 244.15 & 256.42 & 13.07  & 13.40 & 13.37 \\ \bottomrule
\end{tabular}
\end{table*}

Table~\ref{tbl:problem_size} shows the problem size of our PV sequence model~\eqref{obj:PVS}--\eqref{con2:PVS} for some $(n,m)$ settings of the PV sequence. 
Here, the ``\#Vars'' column shows the number of decision variables (i.e., $(m+1)^n$), and the subsequent columns show the number of constraints in Eq.~\eqref{con1:PVS} for the three cases mentioned above. 

The number of constraints grew rapidly as $n$ and $m$ increased in the enumeration case. 
In contrast, the number of constraints was always kept smallest by the transitive reduction among the three cases. 
When $(n,m) = (5,6)$, for instance, transitive reductions reduced the number of constraints in the operation case to $63798/195510 \approx 32.6\%$ for SeqUM and $85272/144060 \approx 59.2\%$ for SeqUS. 

The number of constraints was larger for SeqUM than for SeqUS in the enumeration and operation cases. 
In contrast, the number of constraints was often smaller for SeqUM than for SeqUS in the reduction case. 
Thus, the transitive reduction had a greater impact on SeqUM than on SeqUS in terms of the number of constraints. 

Table~\ref{tbl:comp_time} lists the computation times required for solving the optimization problem~\eqref{obj:PVS}--\eqref{con2:PVS} for some $(n,m)$ settings of the PV sequence. 
Here, ``OM'' indicates that computation was aborted due to a lack of memory. 
The enumeration case often caused out-of-memory errors because of the huge number of constraints (see Table~\ref{tbl:problem_size}), but
the operation and reduction cases completed the computations for all $(n,m)$ settings for the PV sequence. 
Moreover, the transitive reduction made computations faster. 
A notable example is SeqUM with $(n,m) = (5,6)$, for which the computation time in the reduction case (86.02~s) was only one-tenth of that in the operation case (906.76~s). 
These results demonstrate that transitive reduction improves efficiency in terms of both computation time and memory usage.


Table~\ref{tbl:perform} shows the computational performance of our optimization model~\eqref{obj:PVS}--\eqref{con2:PVS} for some $(n,m)$ settings of PV sequences. 
Here, for each $n \in \{3,4\ldots,9\}$, the largest $m$ was chosen such that the computation finished within 30~min. 
Both SeqUM and SeqUS always delivered higher F1 scores than SeqEmp did. 
This indicates that our monotonicity constraint~\eqref{con1:PVS} works well for improving prediction performance. 
The F1 scores provided by SeqUM and SeqUS were very similar and largest with $(n,m) = (7,3)$. 
In light of these results, we use the setting $(n,m) \in \{(7,3),(5,6)\}$ in the following sections.

\subsection{Prediction performance of our PV sequence model}
\label{sec:6.4}

\begin{figure*}[h]
\tabcolsep = 5pt
\centering
\begin{tabular}{ccc}
\includegraphics[keepaspectratio, scale=0.42, bb=50 0 450 330]{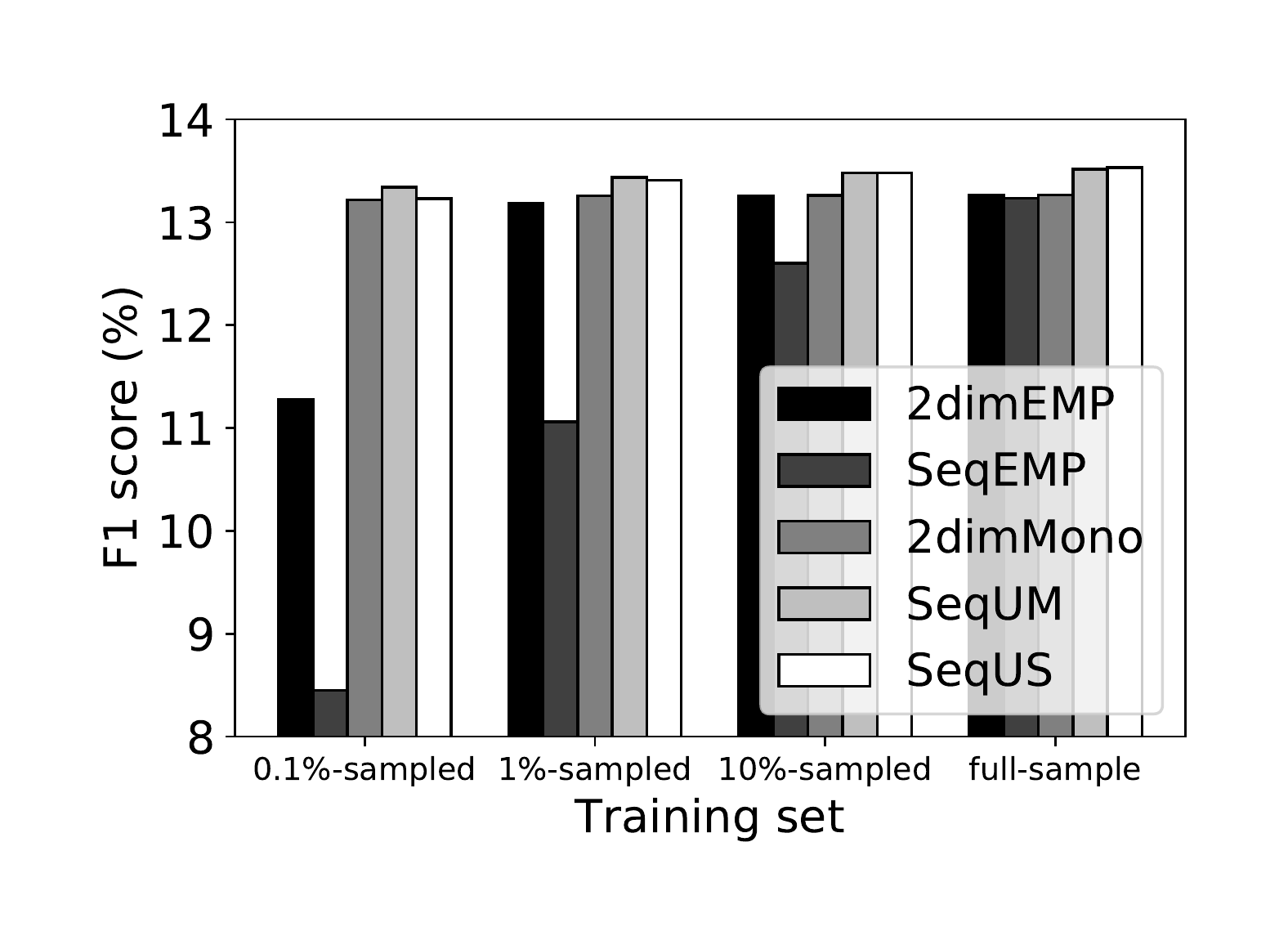}
&\includegraphics[keepaspectratio, scale=0.42, bb=50 0 450 330]{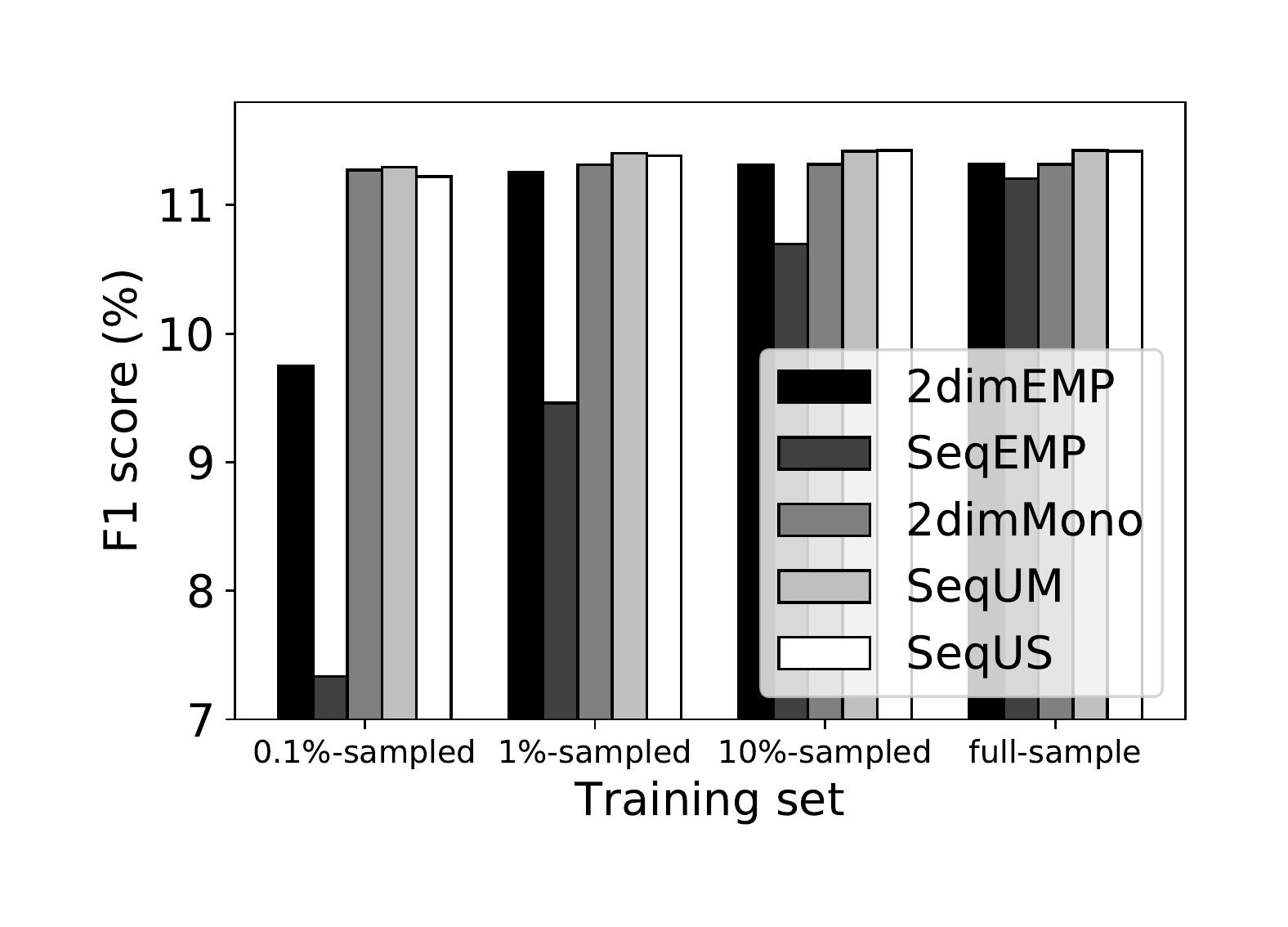}
&\includegraphics[keepaspectratio, scale=0.42, bb=50 0 450 330]{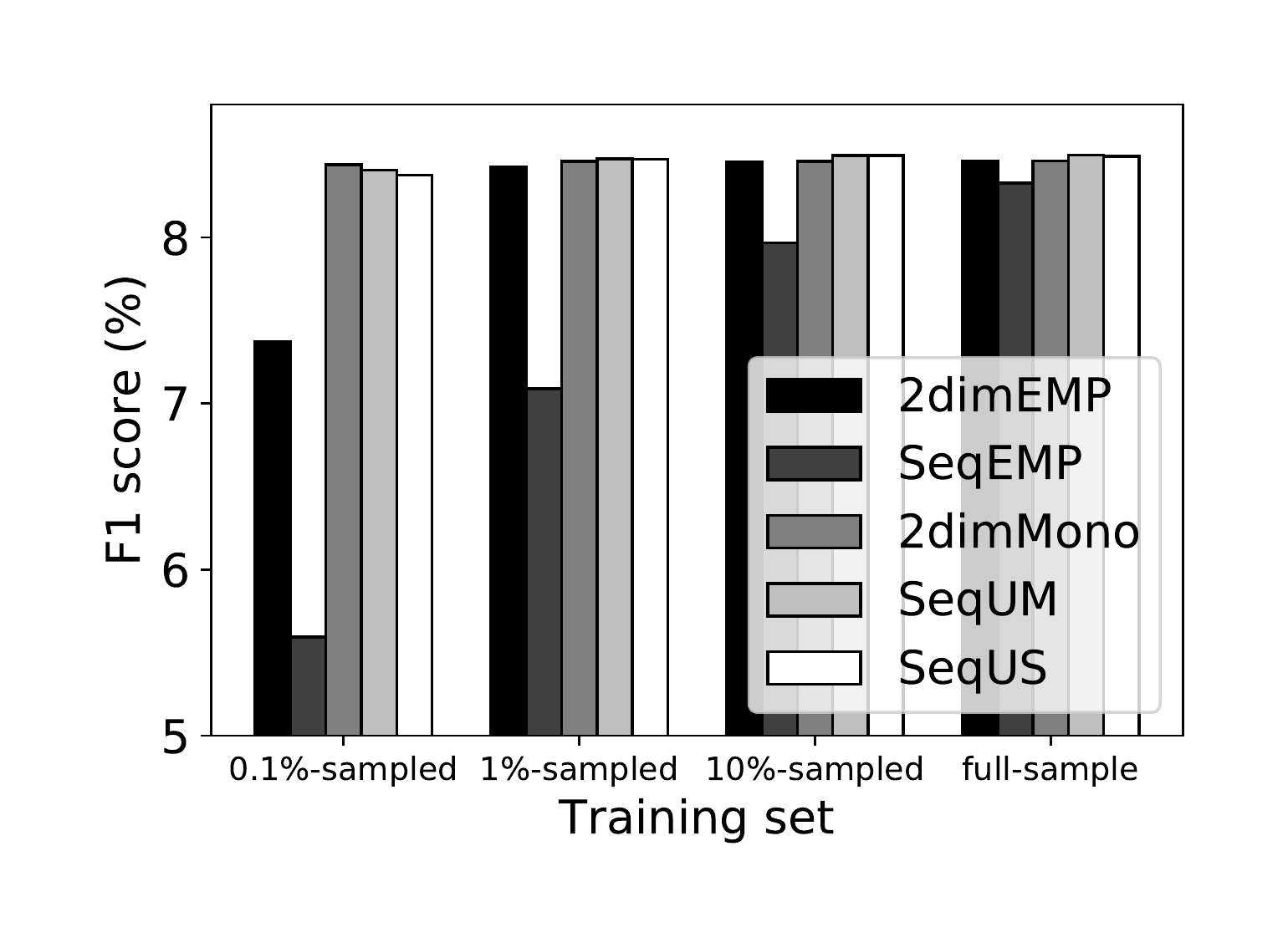}
\\[-.1in]
(a)~$N = 3$, $(n,m)=(7,3)$
&(b)~$N = 5$, $(n,m)=(7,3)$
&(c)~$N = 10$, $(n,m)=(7,3)$\\[-.00in]
\includegraphics[keepaspectratio, scale=0.42, bb=50 0 450 330]{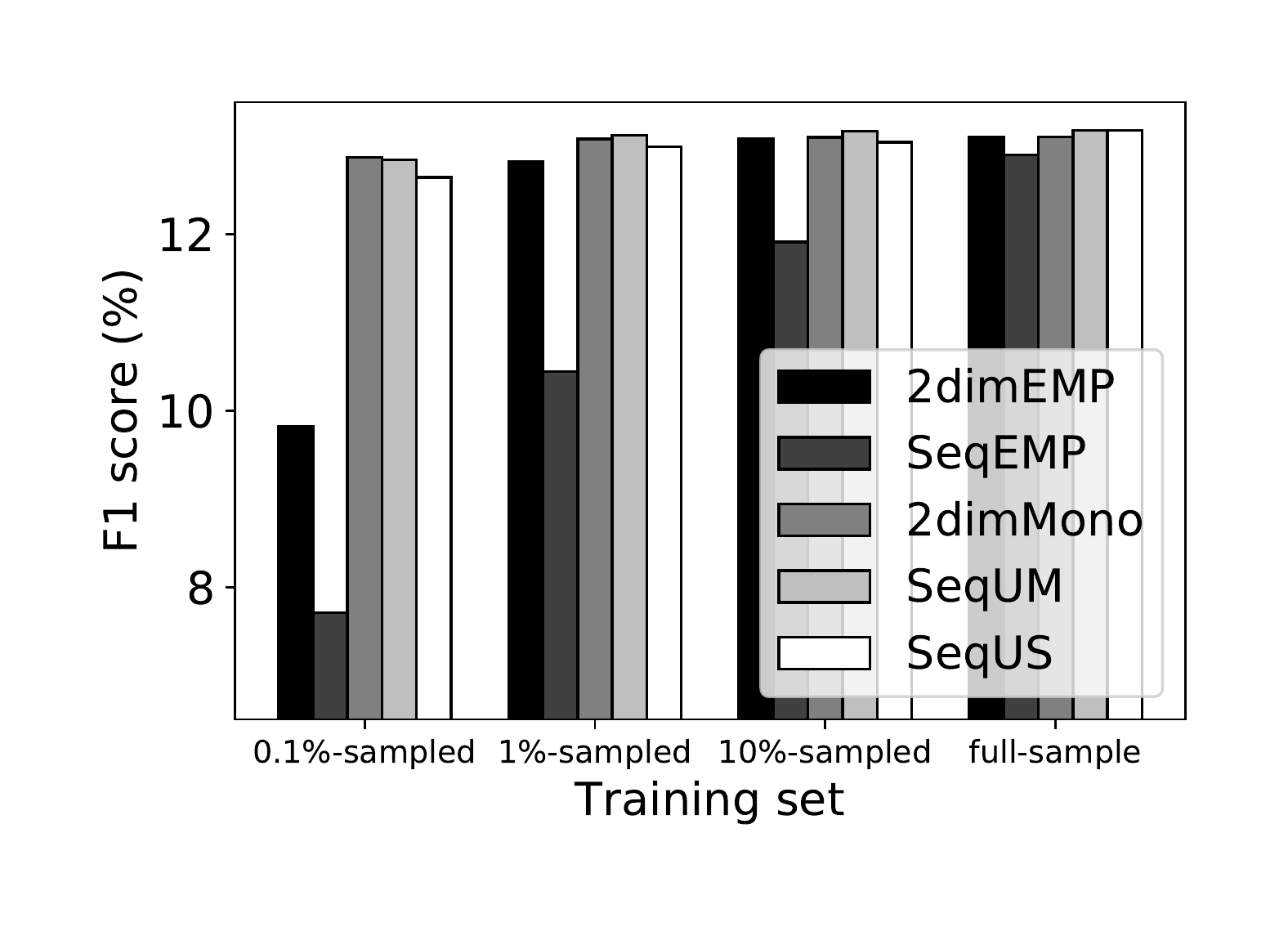}
&\includegraphics[keepaspectratio, scale=0.42, bb=50 0 450 330]{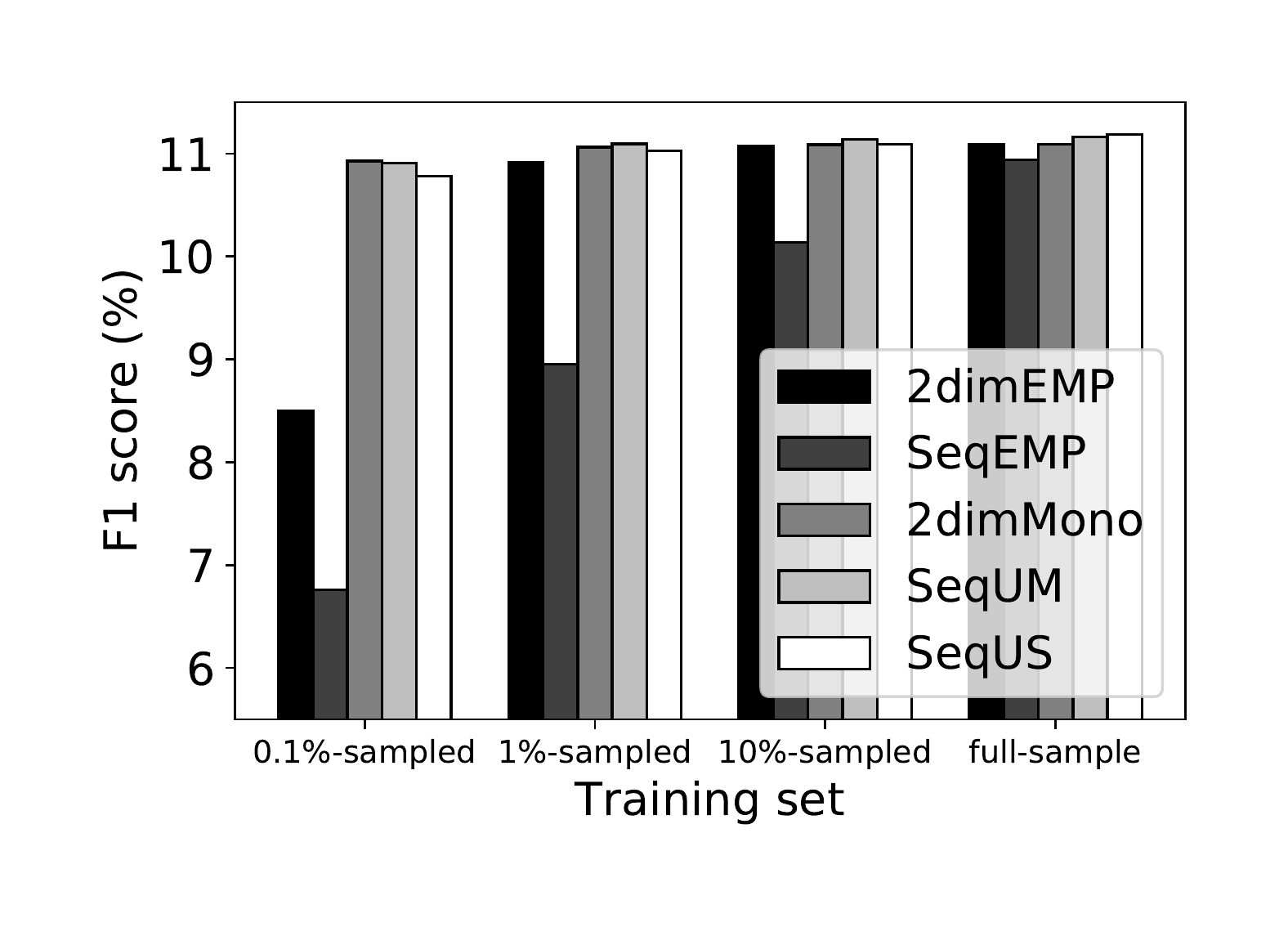}
&\includegraphics[keepaspectratio, scale=0.42, bb=50 0 450 330]{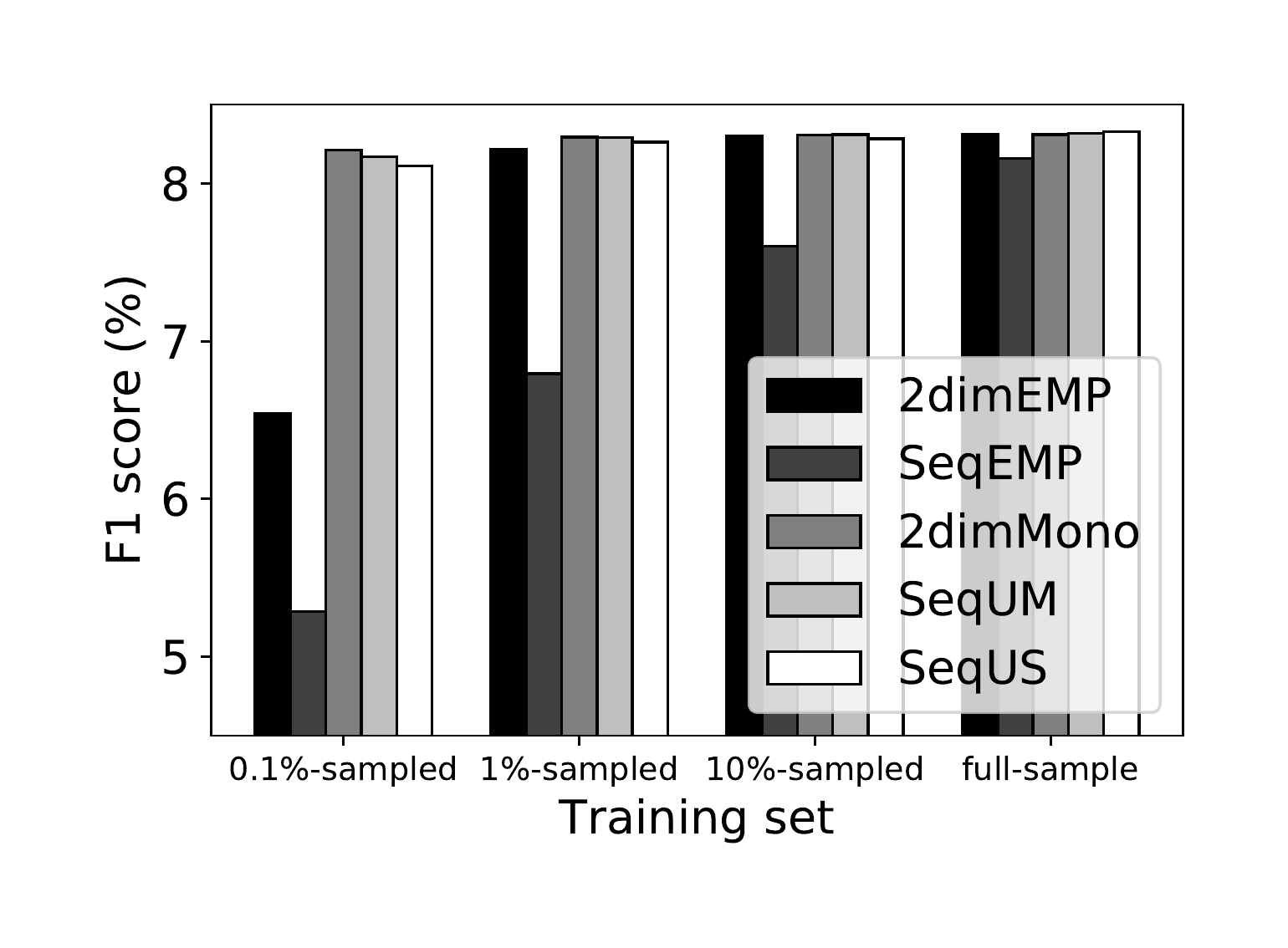}
\\[-.1in]
(d)~$N = 3$, $(n,m)=(5,6)$
&(e)~$N = 5$, $(n,m)=(5,6)$
&(f)~$N = 10$, $(n,m)=(5,6)$
\end{tabular}
\caption{Comparison of prediction performance versus the two-dimensional probability table.}
\label{fig:hasse12_sampling}
\end{figure*}
\begin{figure*}[t]
\centering
\tabcolsep = 5pt
\begin{tabular}{ccc}
\includegraphics[keepaspectratio, scale=0.39, bb=40 20 470 410]{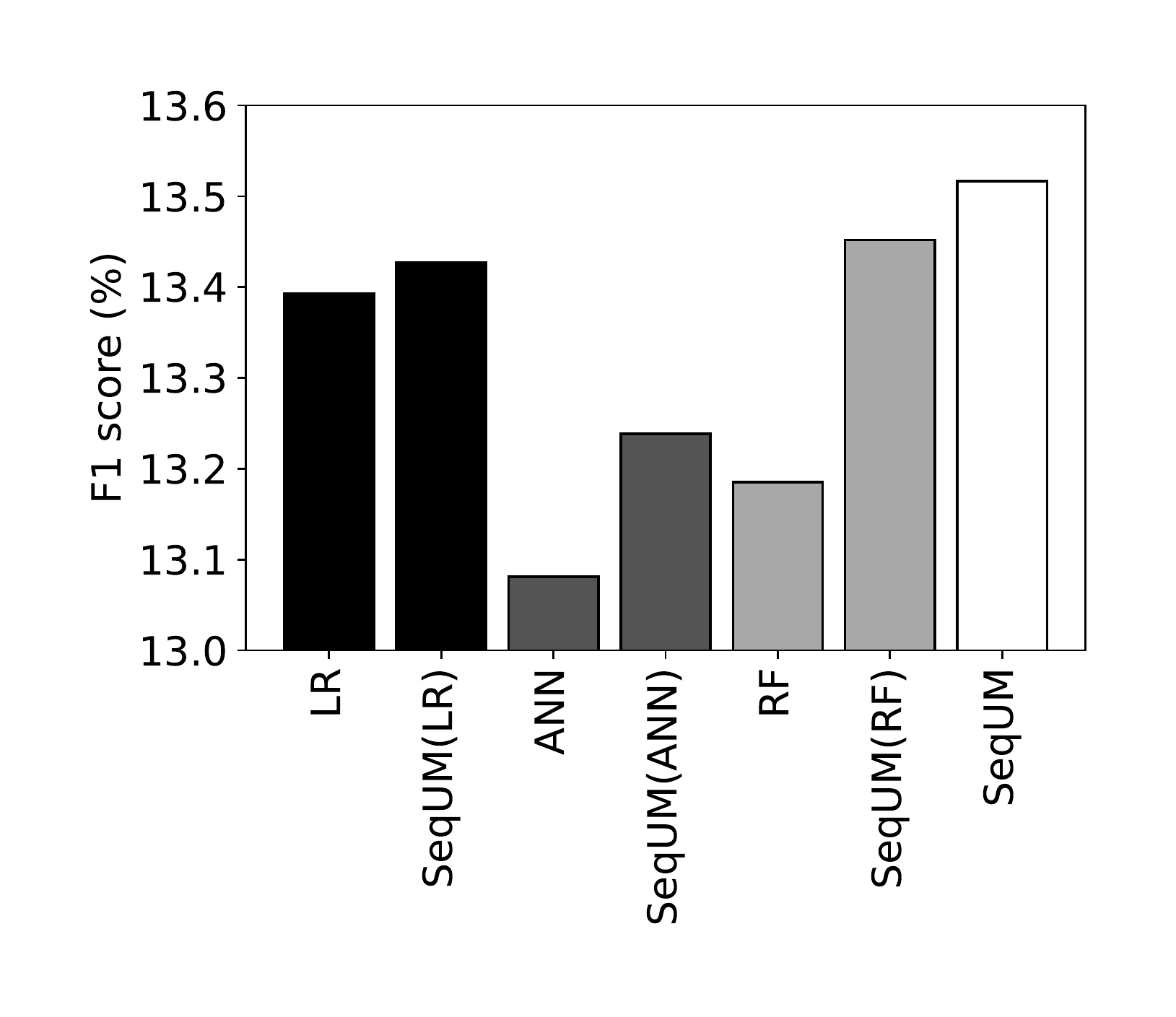}& 
\includegraphics[keepaspectratio, scale=0.39, bb=40 20 470 410]{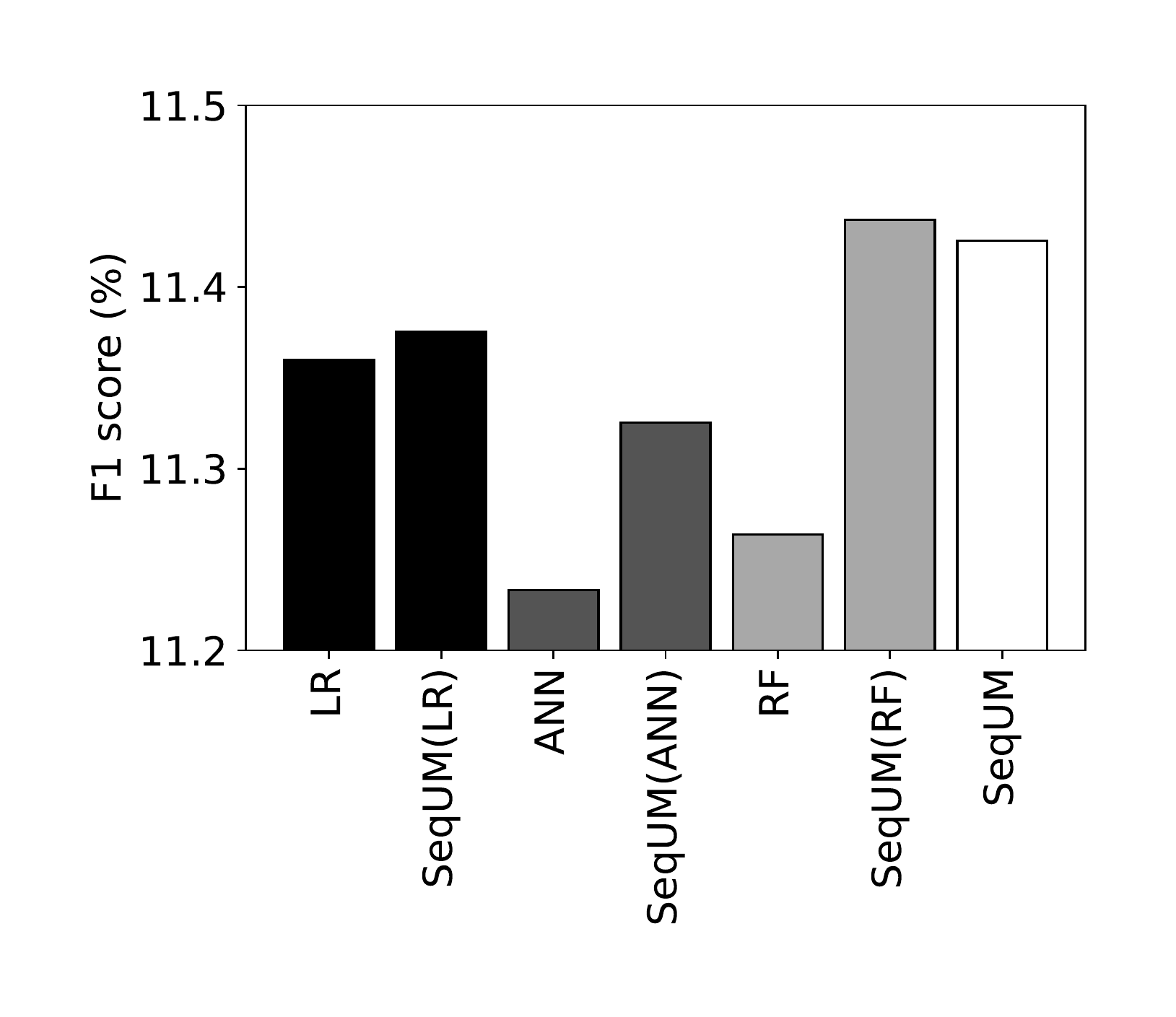}& 
\includegraphics[keepaspectratio, scale=0.39, bb=40 20 470 410]{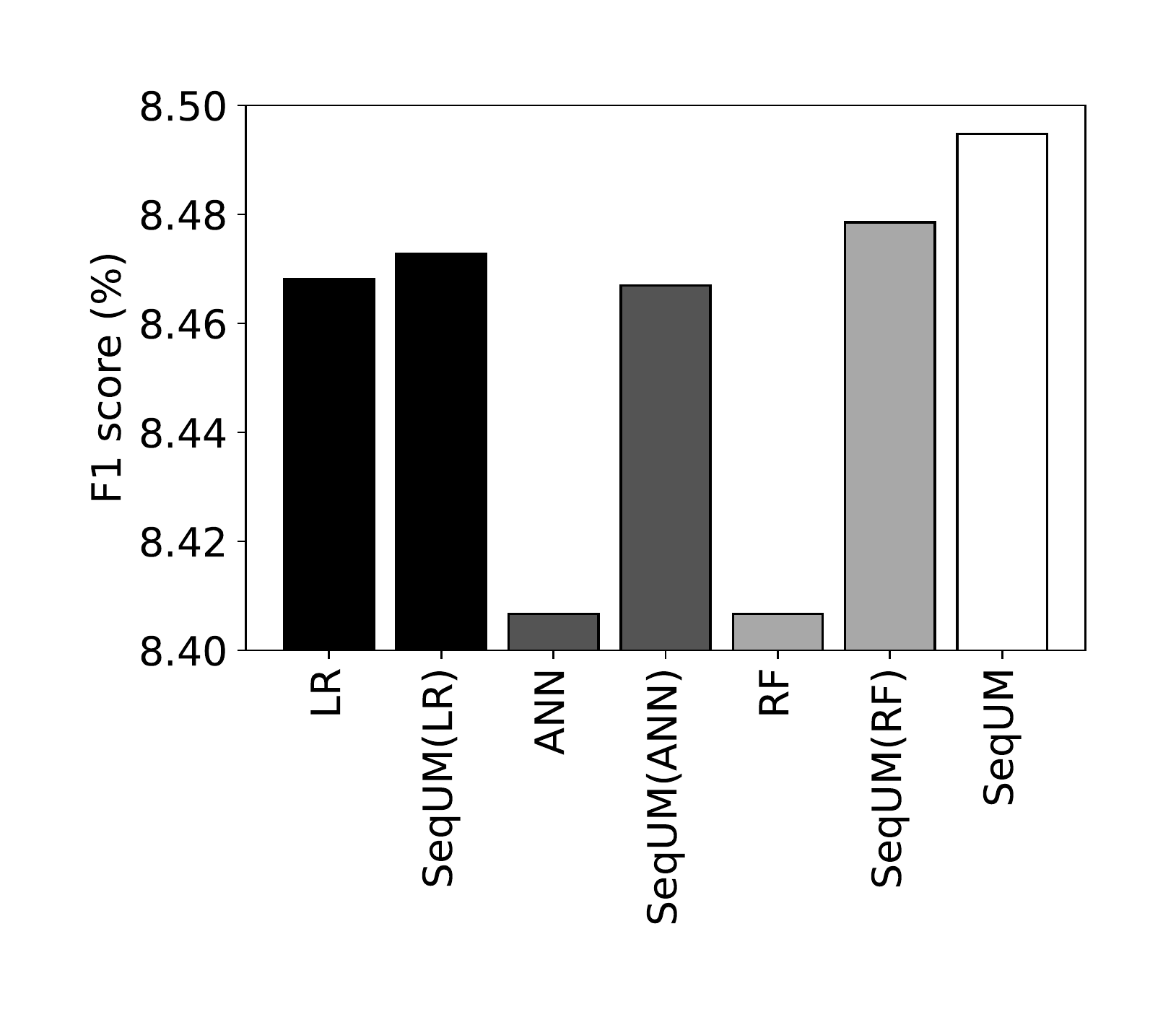}
\\[-.05in]
(a) $N = 3$, $(n,m)=(7,3)$&
(b) $N = 5$, $(n,m)=(7,3)$&
(c) $N = 10$, $(n,m)=(7,3)$\\[-.00in]
\includegraphics[keepaspectratio, scale=0.39, bb=40 20 470 410]{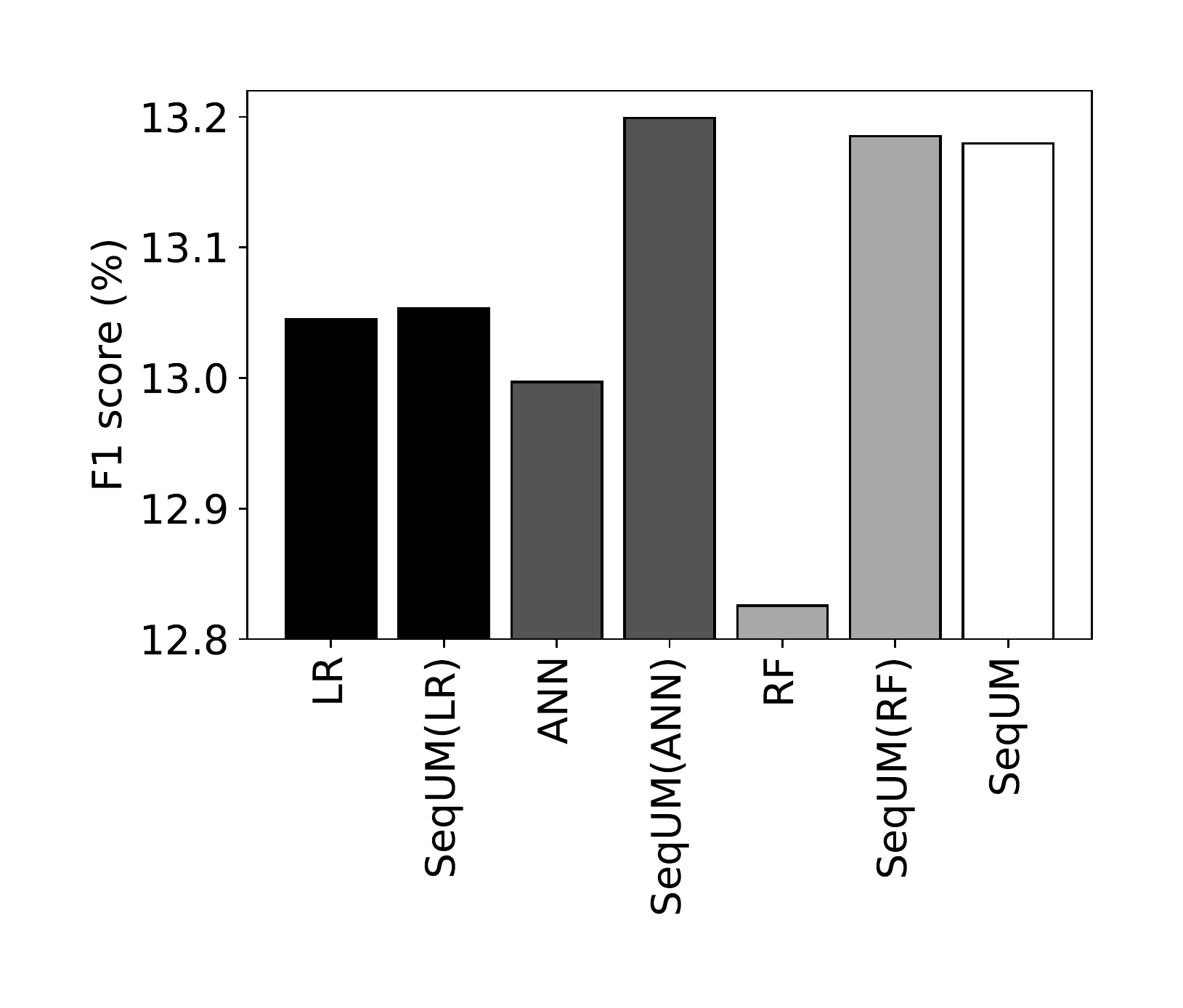}&
\includegraphics[keepaspectratio, scale=0.39, bb=40 20 470 410]{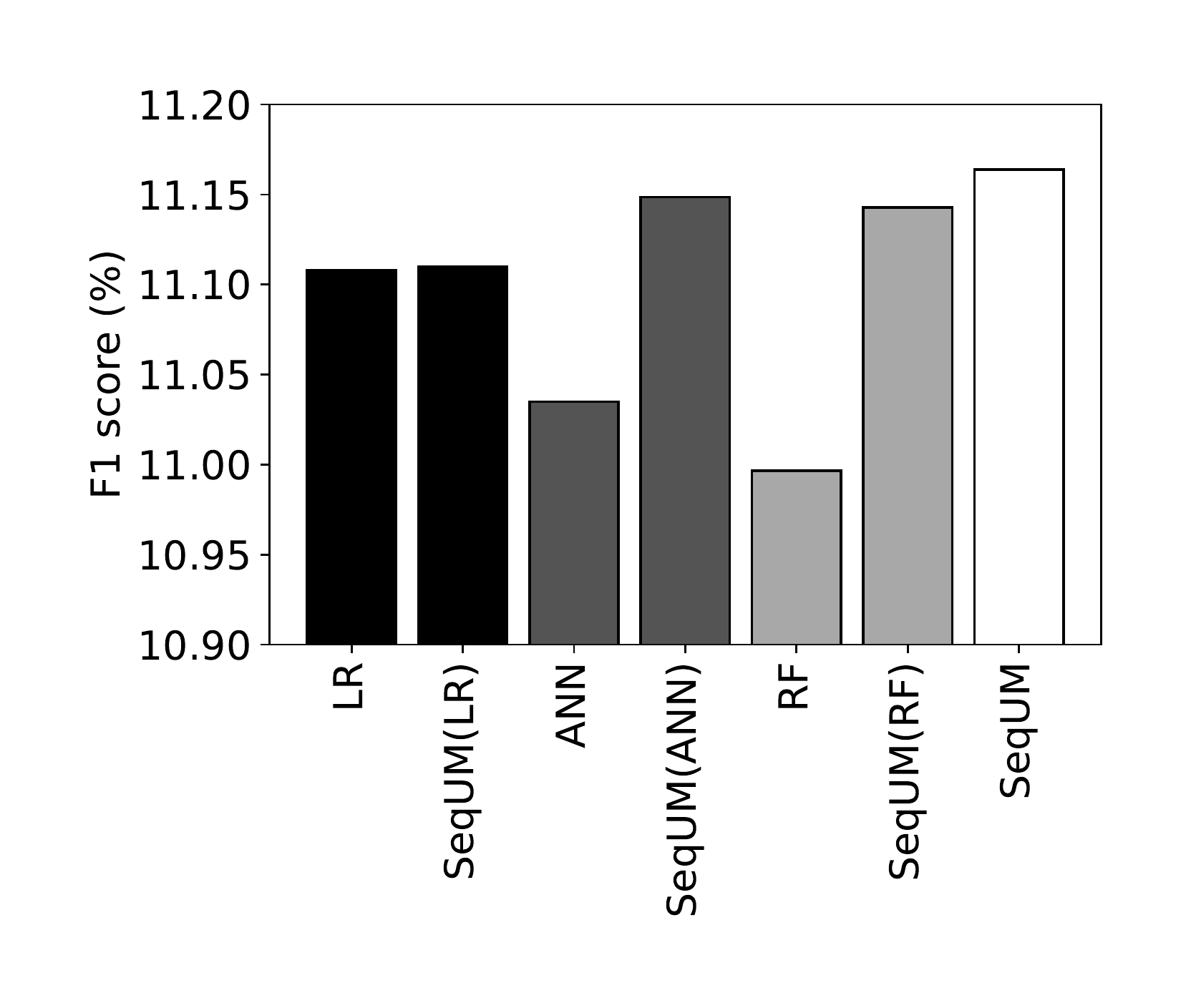}&
\includegraphics[keepaspectratio, scale=0.39, bb=40 20 470 410]{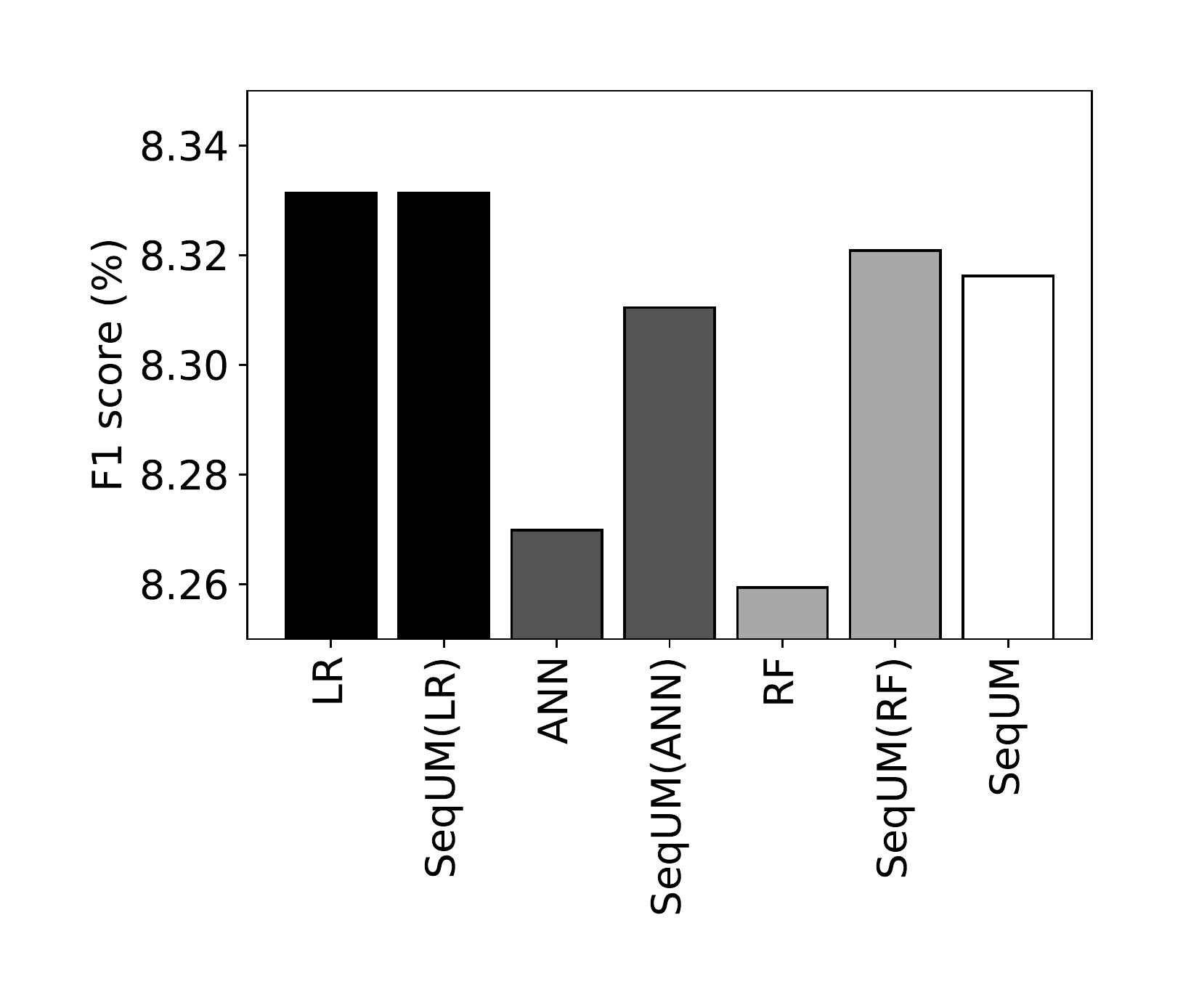}
\\[-.05in]
(d) $N = 3$, $(n,m)=(5,6)$&
(e) $N = 5$, $(n,m)=(5,6)$&
(f) $N = 10$, $(n,m)=(5,6)$
\end{tabular}
\caption{Comparison of prediction performance versus machine learning methods.}
\label{fig:ml_hasse}
\end{figure*}

\begin{figure*}[t]
\centering
\tabcolsep = 10pt
\begin{tabular}{ccc}
\includegraphics[width=5.50cm,bb=70 0 400 250,clip]{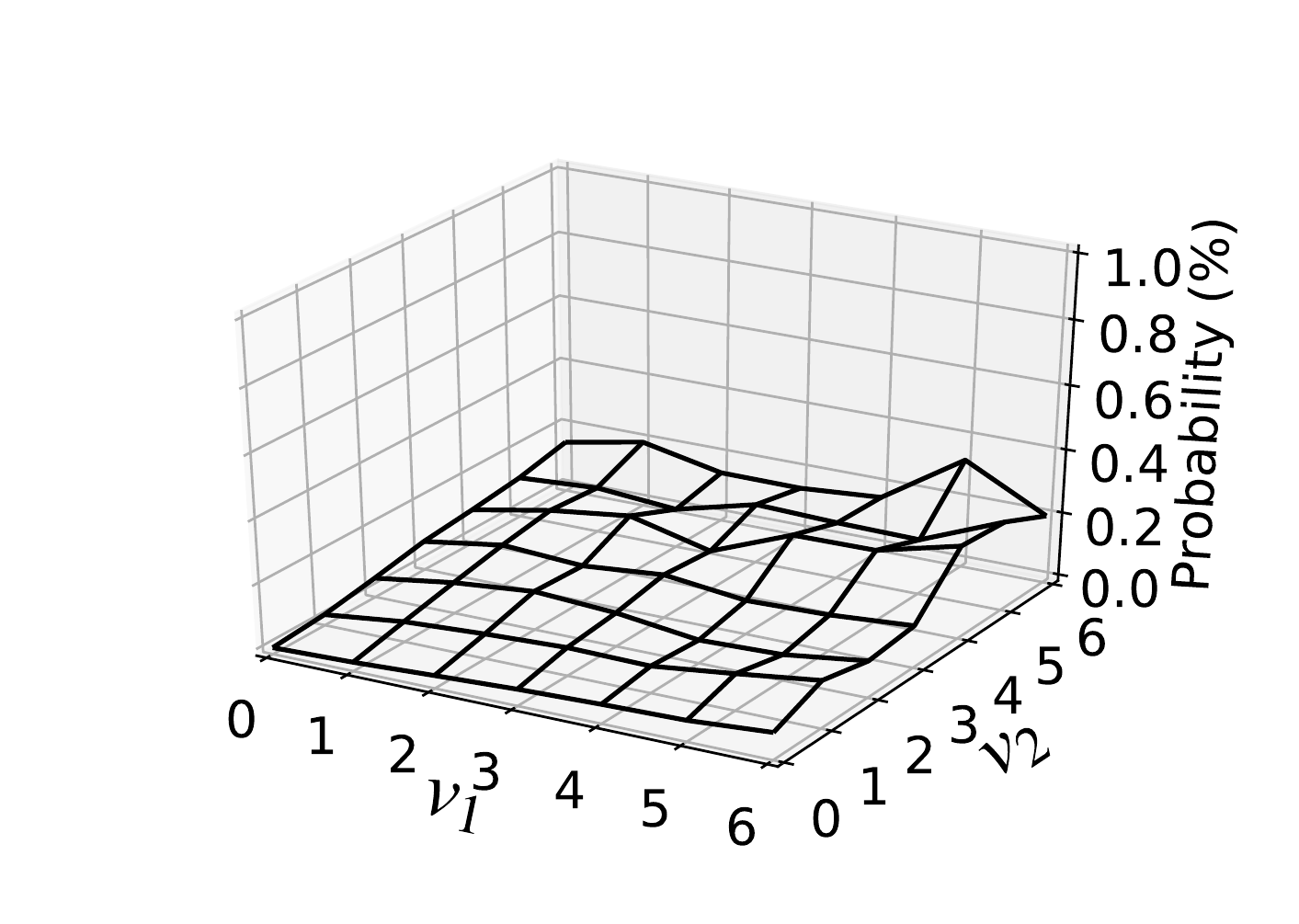} & 
\includegraphics[width=5.50cm,bb=70 0 400 250,clip]{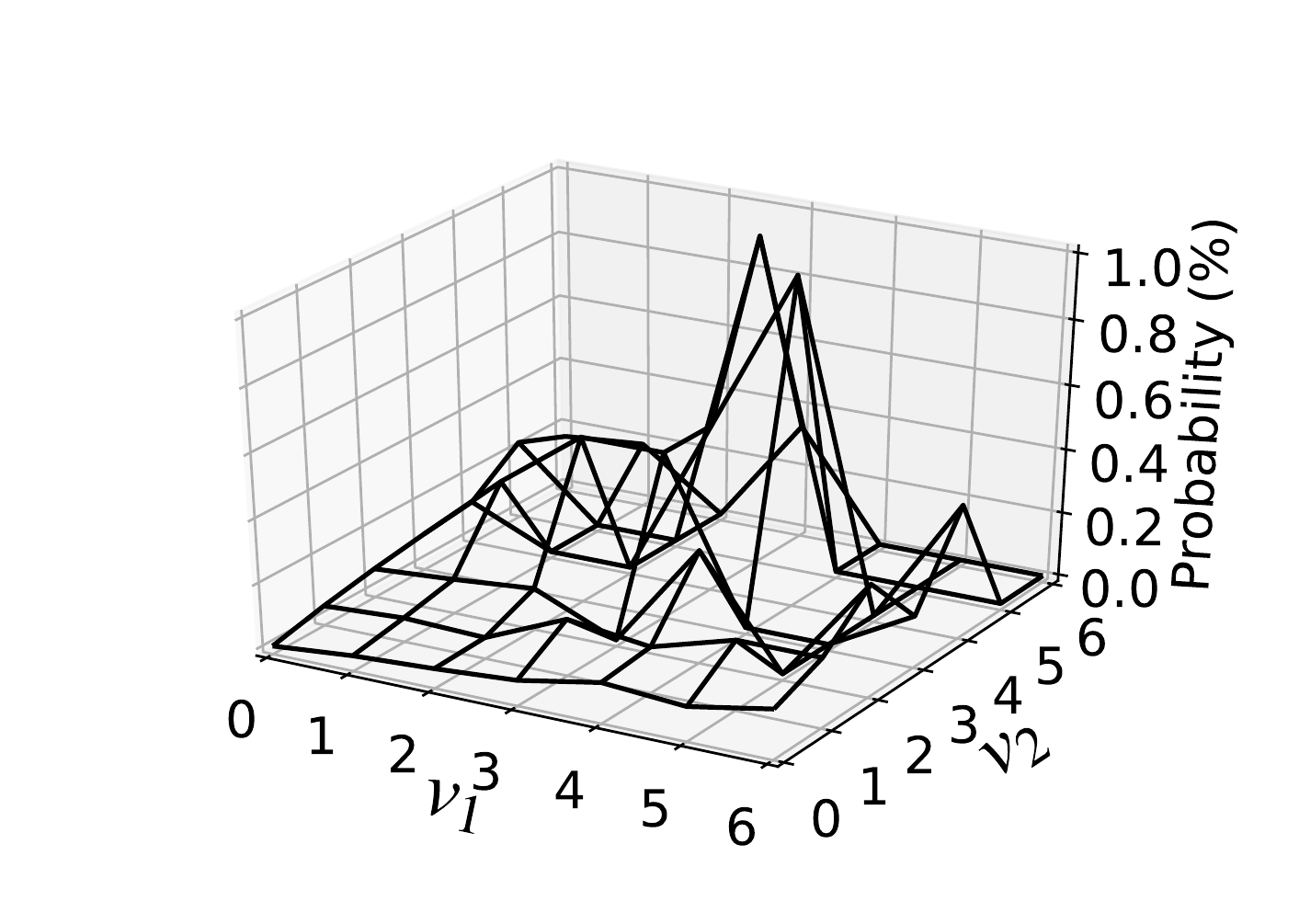} & 
\includegraphics[width=5.50cm,bb=70 0 400 250,clip]{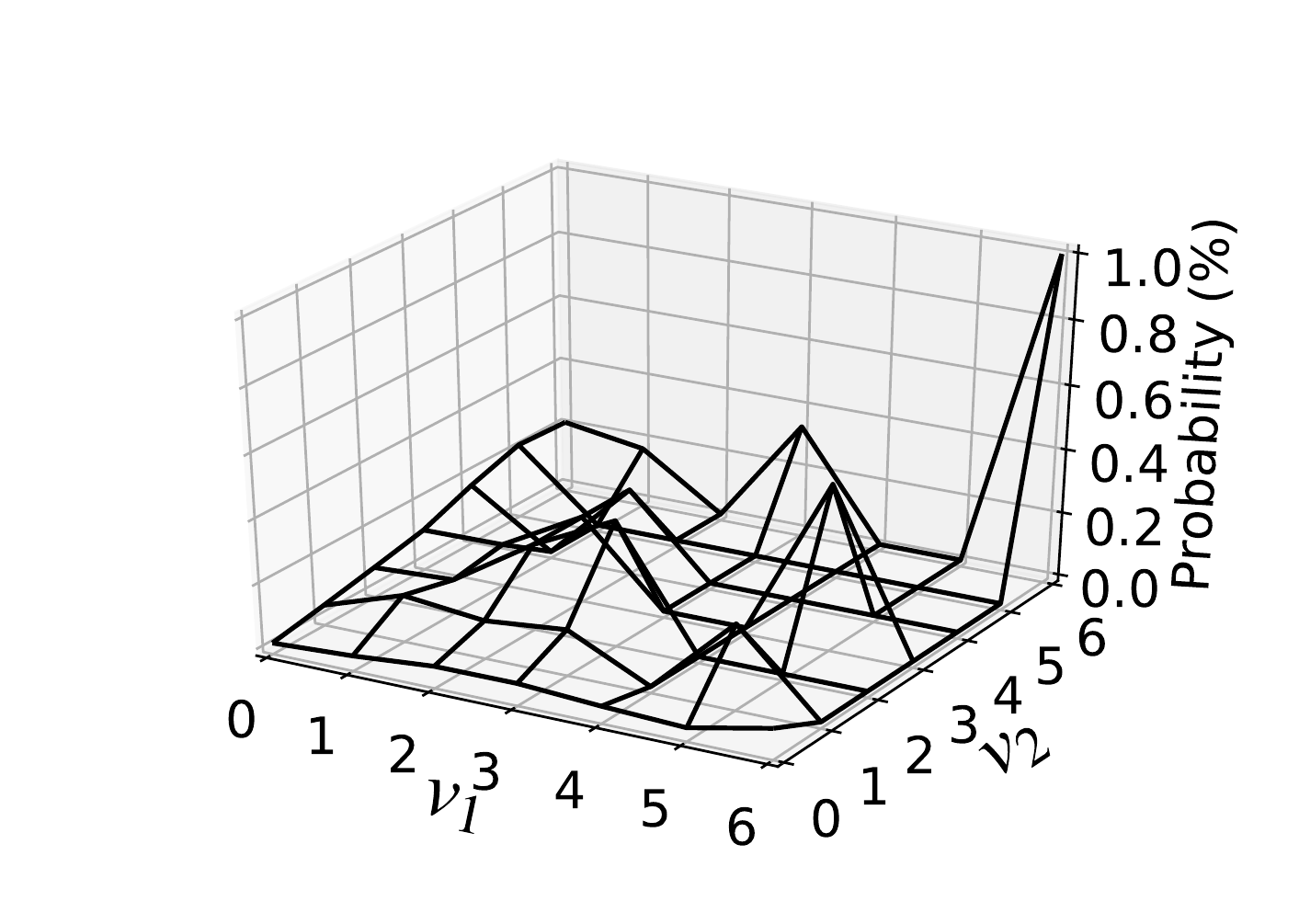}
\\[-.1in]
(a) SeqEmp, $v_3=0$ & 
(b) SeqEmp, $v_3=1$ & 
(c) SeqEmp, $v_3=2$
\\[.1in]
\includegraphics[width=5.50cm,bb=70 0 400 250,clip]{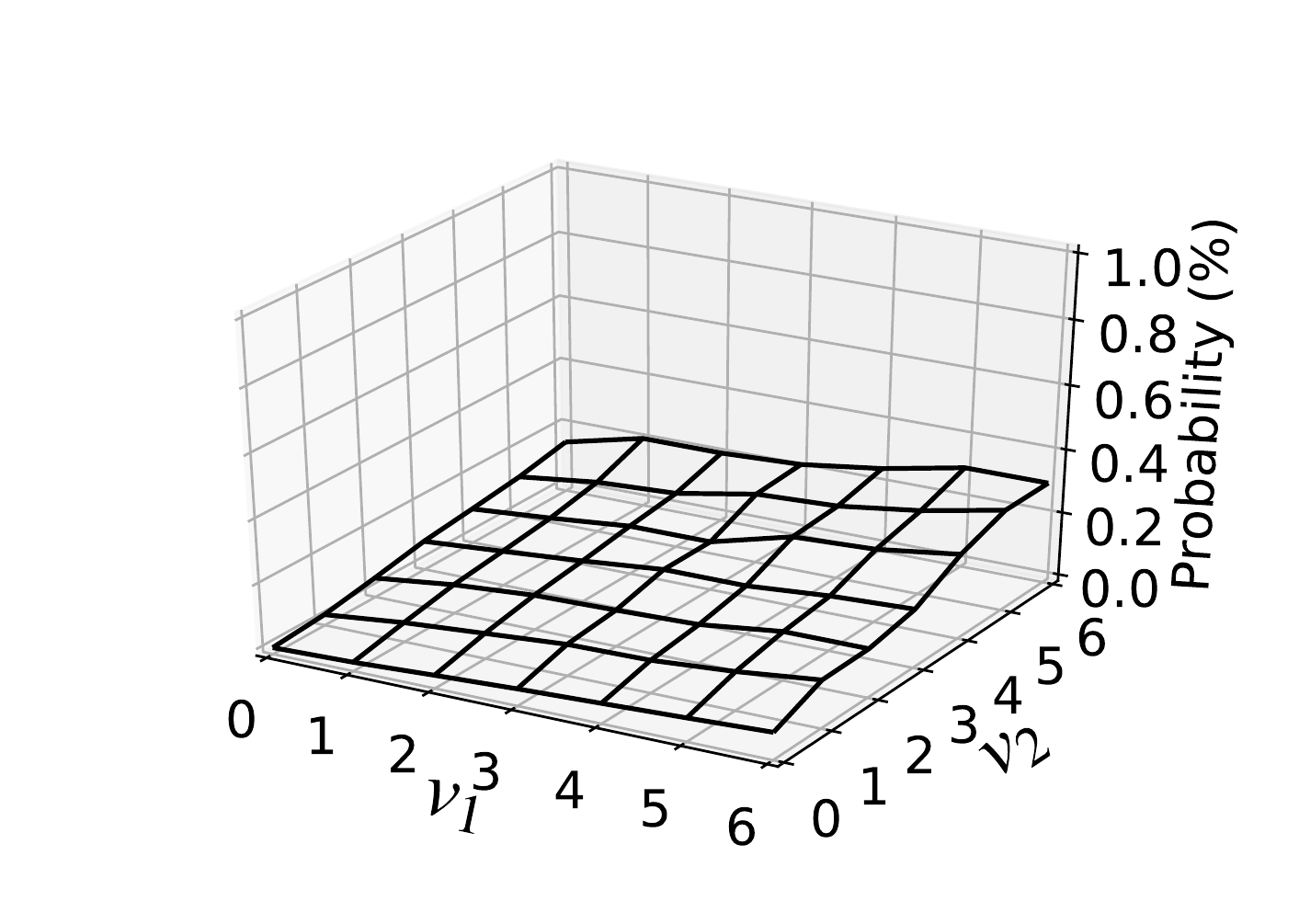} &
\includegraphics[width=5.50cm,bb=70 0 400 250,clip]{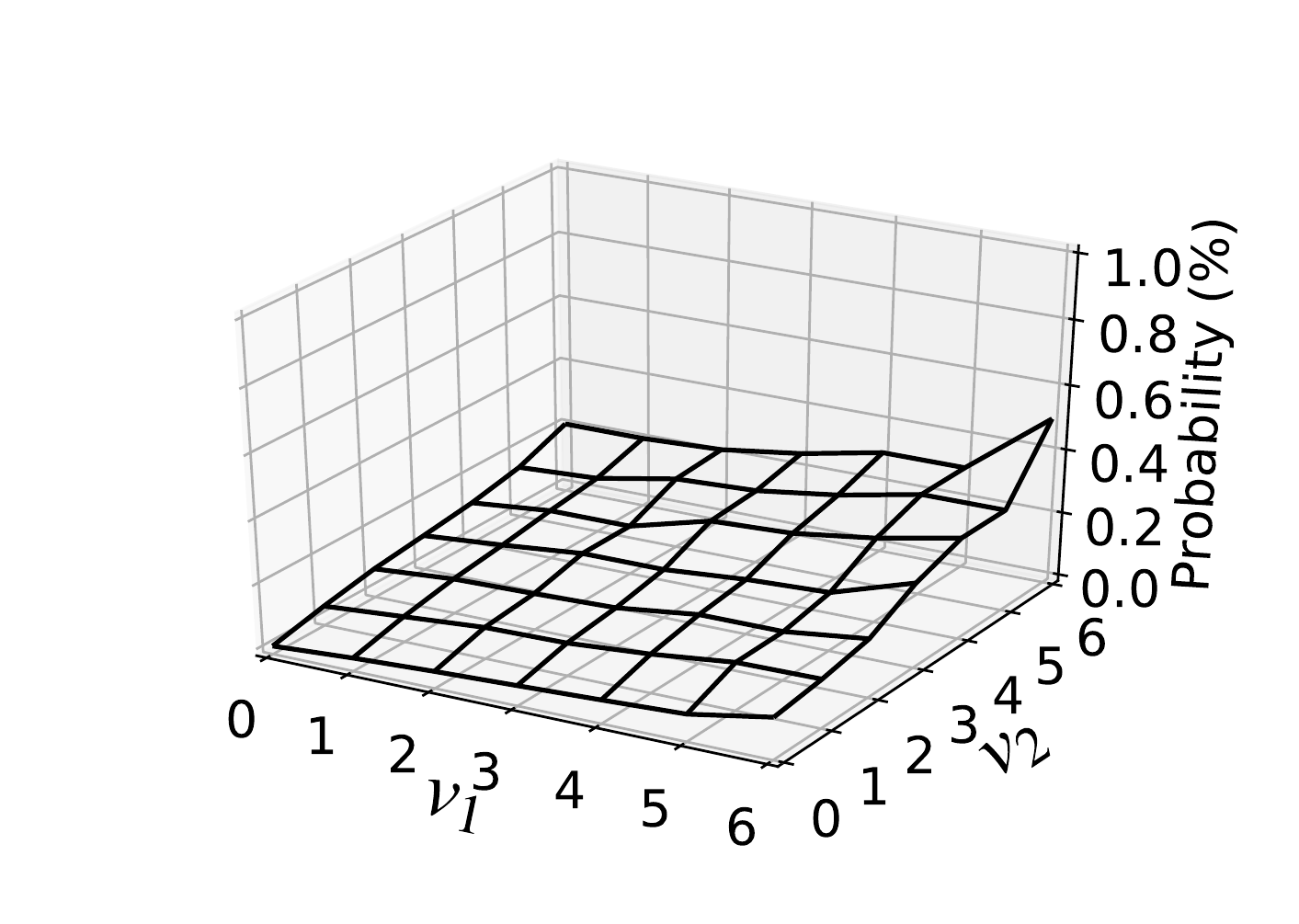} &
\includegraphics[width=5.50cm,bb=70 0 400 250,clip]{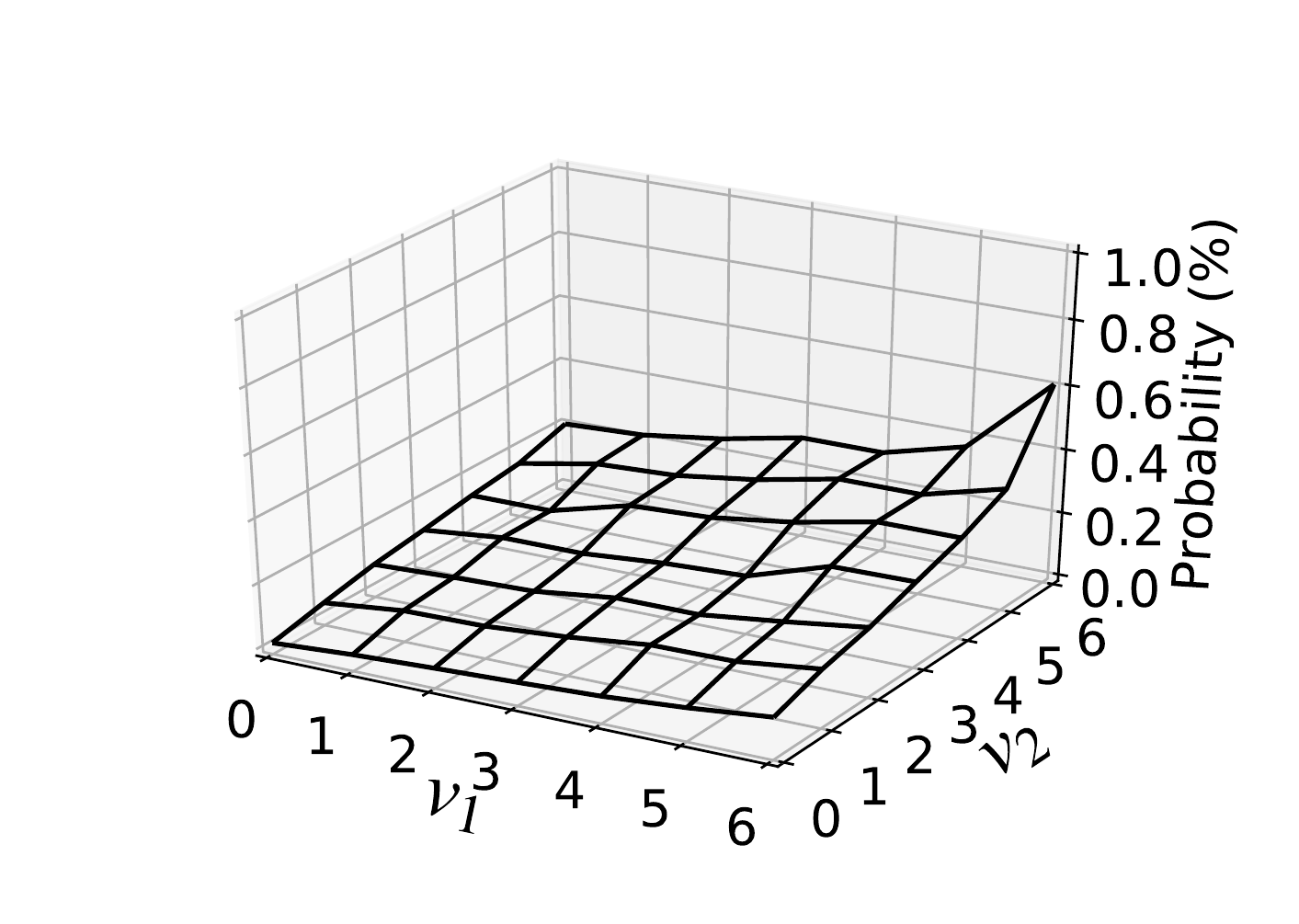}
\\[-.1in]
(d) SeqUM, $v_3=0$ & 
(e) SeqUM, $v_3=1$ & 
(f) SeqUM, $v_3=2$
\\[.1in]
\includegraphics[width=5.50cm,bb=70 0 400 250,clip]{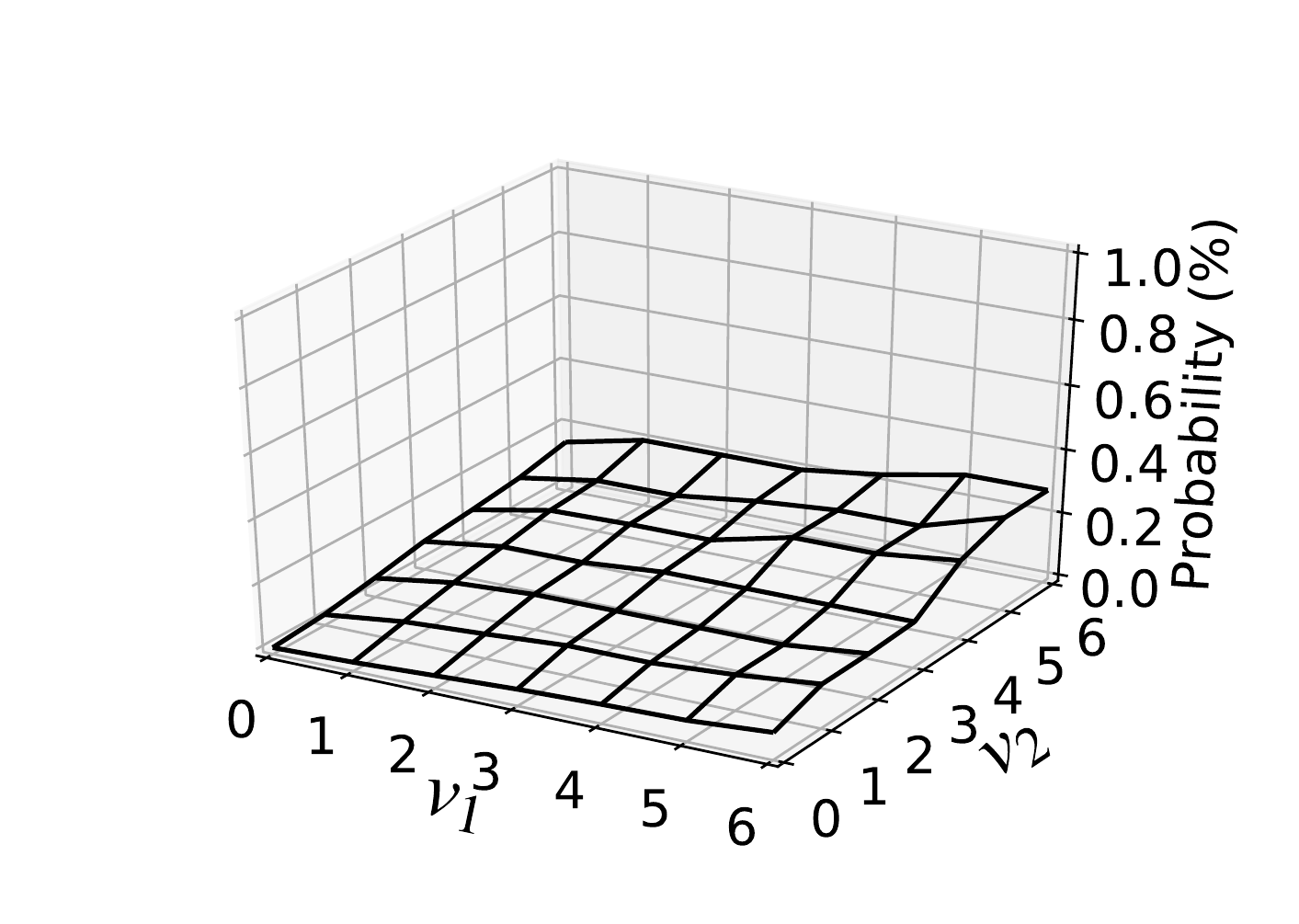} & 
\includegraphics[width=5.50cm,bb=70 0 400 250,clip]{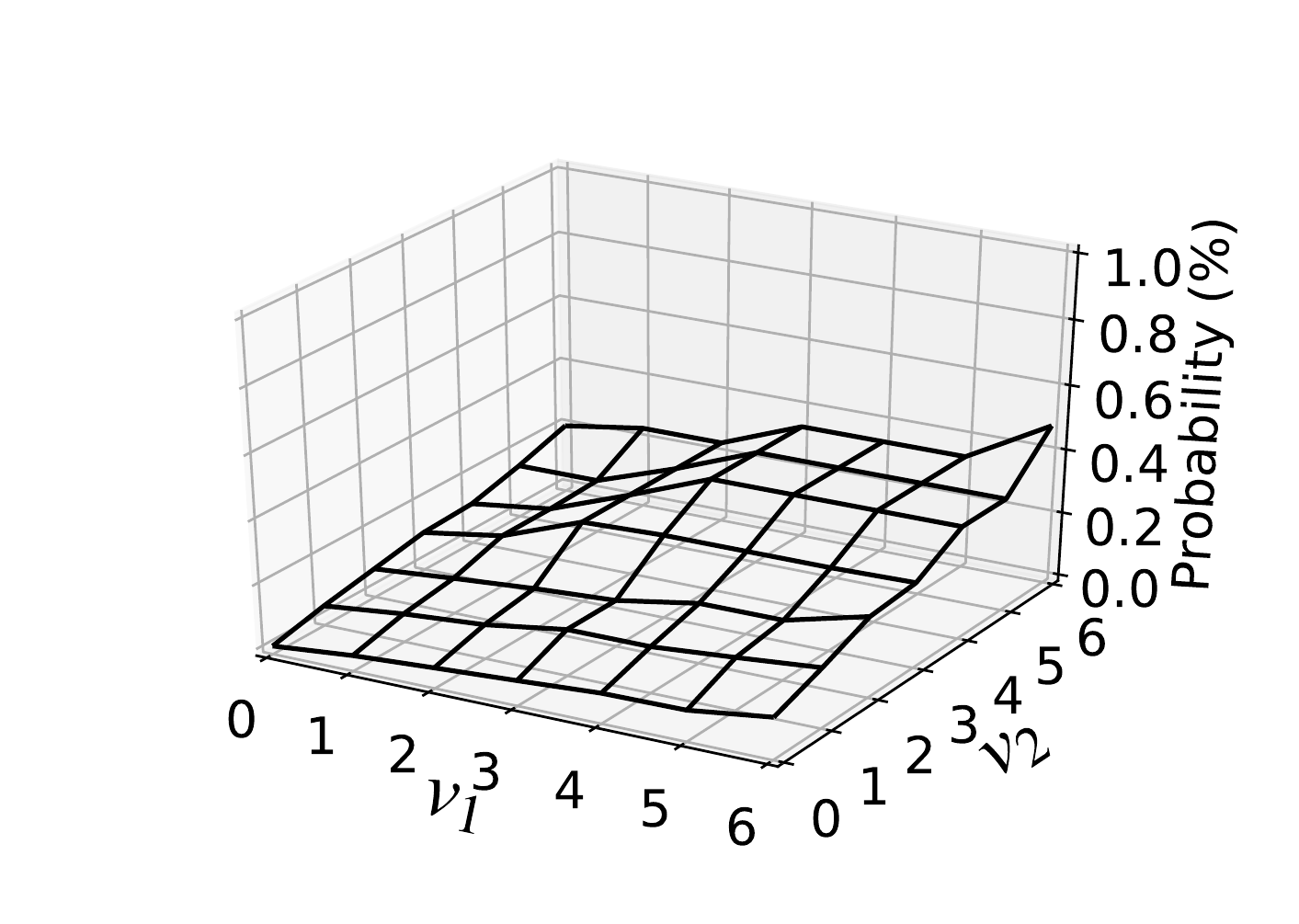} & 
\includegraphics[width=5.50cm,bb=70 0 400 250,clip]{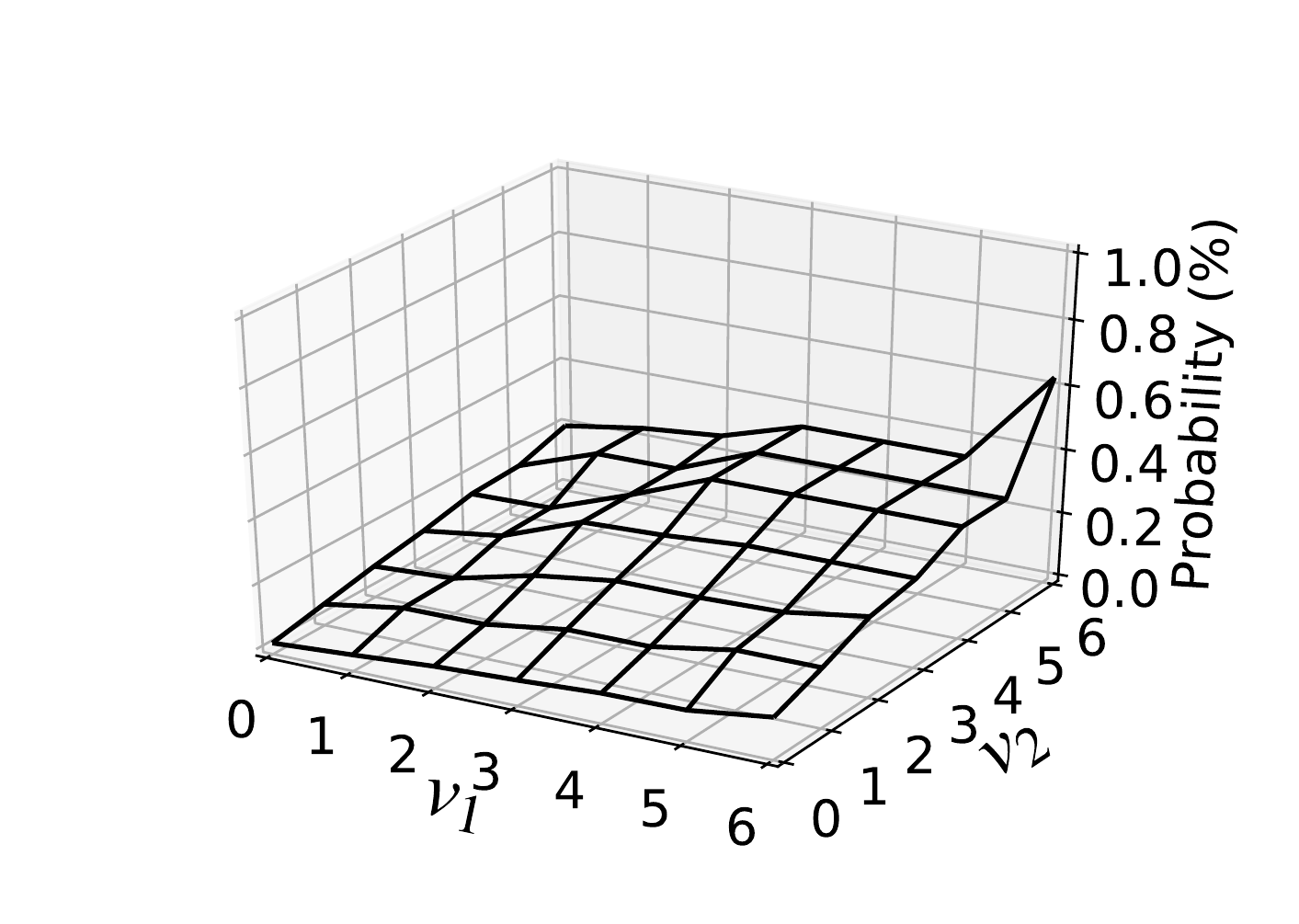}
\\[-.1in]
(g) SeqUS, $v_3=0$ & 
(h) SeqUS, $v_3=1$ & 
(i) SeqUS, $v_3=2$
\end{tabular}
\caption{Item-choice probabilities estimated from the full-sample training set with $(n,m)=(5,6)$.}
\label{fig:seq_visualize_full}
\end{figure*}

Fig.~\ref{fig:hasse12_sampling} shows F1 scores of the two-dimensional probability table and our PV sequence model using the sampled training sets, where the number of selected items is $N \in \{3,5,10\}$, and the setting of the PV sequence is $(n,m) \in \{(7,3),(5,6)\}$.

When the full-sample training set was used, SeqUM and SeqUS always delivered better prediction performance than the other methods did. 
When the 1\%- and 10\%-sampled training sets were used, the prediction performance of SeqUS decreased slightly, whereas SeqUM still performed best among all the methods. 
When the 0.1\%-sampled training set was used, 2dimMono always performed better than SeqUS did, and  2dimMono also had the best prediction performance in the case $(n,m)=(5,6)$.
These results suggest that our PV sequence model performs very well, especially when the sample size is sufficiently large. 

The prediction performance of SeqEmp deteriorated rapidly as the sampling rate decreased, and this performance was always much worse than that of 2dimEmp. 
Meanwhile, SeqUM and SeqUS maintained high prediction performance even when the 0.1\%-sampled training set was used. 
This suggests that the monotonicity constraint~\eqref{con1:PVS} in our PV sequence model is more effective than the monotonicity constraints~\eqref{con1:Mono}--\eqref{con2:Mono} in the two-dimensional monotonicity model. 

Fig.~\ref{fig:ml_hasse} shows F1 scores for the machine learning methods (LR, ANN, and RF) and our PV sequence model (SeqUM) using the full-sample training set, where the number of selected items is $N \in \{3,5,10\}$, and the PV sequence setting is $(n,m) \in \{(7,3),(5,6)\}$.
Note that in this figure, SeqUM(\,$\ast$\,) represents the optimization model~\eqref{obj:PVS}--\eqref{con2:PVS}, where the item-choice probabilities computed by each machine learning method were substituted into $(\hat{x}_{\bm{v}})_{\bm{v} \in \Gamma}$ (see Section 4.4).

Prediction performance was better for SeqUM than for all the machine learning methods, except in the case of Fig.~\ref{fig:ml_hasse}(f), where LR showed better prediction performance. 
Moreover, SeqUM(\,$\ast$\,) improved the prediction performance of the machine leaning methods, particularly for ANN and RF. 
This suggests that our monotonicity constraint~\eqref{con1:PVS} is also very helpful in correcting prediction values from other machine learning methods. 

\subsection{Analysis of estimated item-choice probabilities}

Fig.~\ref{fig:seq_visualize_full} shows item-choice probabilities estimated by our PV sequence model using the full-sample training set, where the PV sequence setting is $(n,m) = (5,6)$. 
Here, we focus on PV sequences in the form $\bm{v} = (v_1,v_2,v_3,0,0) \in \Gamma$ and depict estimates of item-choice probabilities on $(v_1,v_2) \in [0,m] \times [0,m]$ for each $v_3 \in \{0,1,2\}$. 
Note also that the number of associated user--item pairs decreased as the value of $v_3$ increased.  


Because SeqEmp does not consider the monotonicity constraint~\eqref{con1:PVS}, item-choice probabilities estimated by SeqEmp have irregular shapes for $v_3 \in \{1,2\}$. 
In contrast, item-choice probabilities estimated with the monotonicity constraint~\eqref{con1:PVS} are relatively smooth.
Because of the \texttt{Up} operation, item-choice probabilities estimated by SeqUM and SeqUS increase as $(v_1,v_2)$ moves from $(0,0)$ to $(6,6)$. 
Because of the \texttt{Move} operation, item-choice probabilities estimated by SeqUM also increase as $(v_1,v_2)$ moves from $(0,6)$ to $(6,0)$. 
Item-choice probabilities estimated by SeqUS are relatively high around $(v_1,v_2) = (3,3)$. 
This highlights the difference in the monotonicity constraint~\eqref{con1:PVS} between posets $(\Gamma,\preceq_{\texttt{UM}})$ and $(\Gamma,\preceq_{\texttt{US}})$. 


Fig.~\ref{fig:seq_visualize_sampled} shows item-choice probabilities estimated by our PV sequence model using the 10\%-sampled training set, where the PV sequence setting is $(n,m) = (5,6)$. 
Since the sample size was reduced in this case, item-choice probabilities estimated by SeqEmp are highly unstable. 
In particular, item-choice probabilities were estimated to be zero for all $(v_1,v_2)$ with $v_1 \ge 3$ in Fig.~\ref{fig:seq_visualize_sampled}(c), but this is unreasonable from the perspective of frequency.   
In contrast, SeqUM and SeqUS estimated item-choice probabilities that increase monotonically with respect to $(v_1,v_2)$.  

\begin{figure*}[t]
\centering
\tabcolsep = 10pt
\begin{tabular}{ccc}
\includegraphics[width=5.50cm,bb=70 0 400 250,clip]{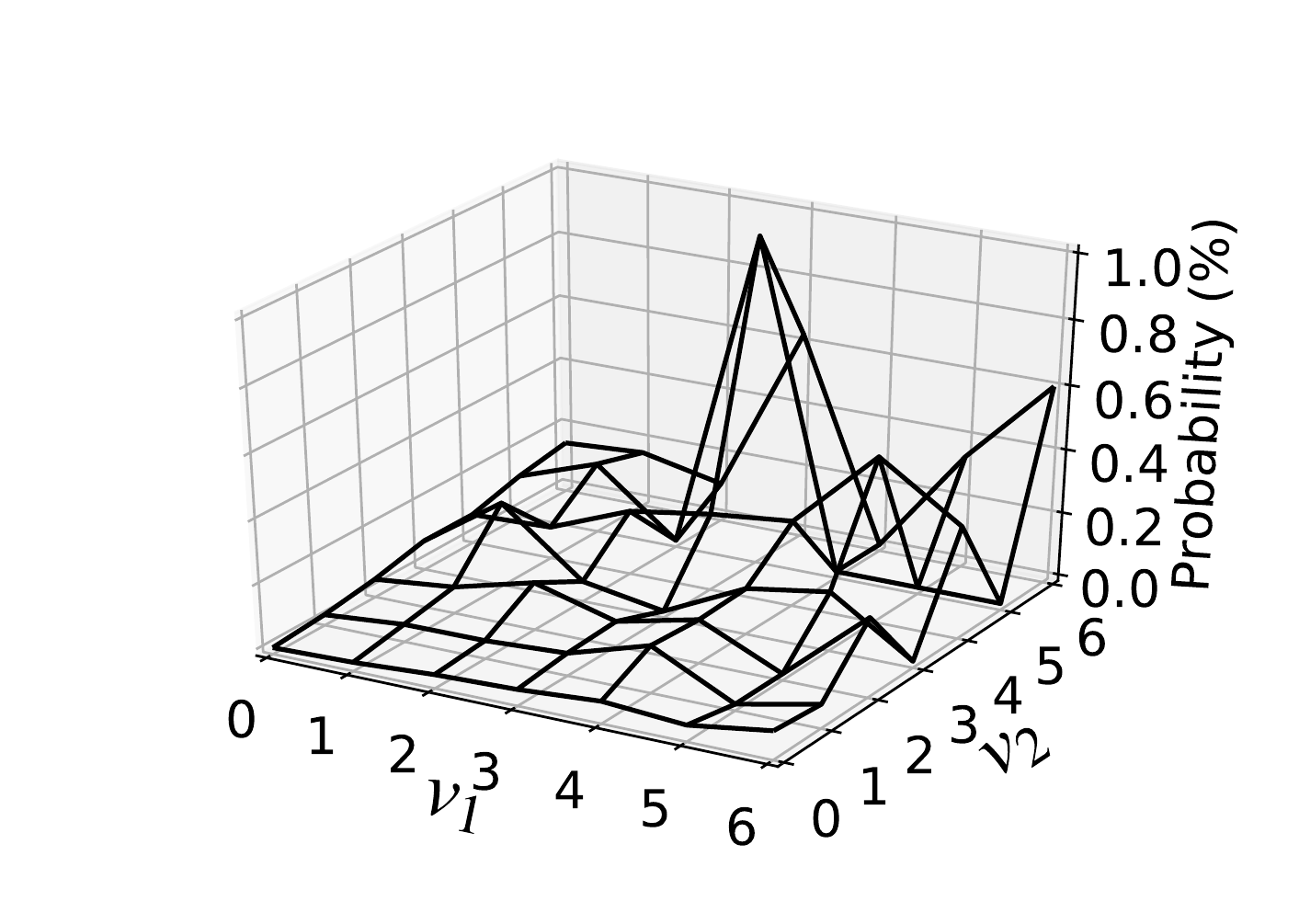} & 
\includegraphics[width=5.50cm,bb=70 0 400 250,clip]{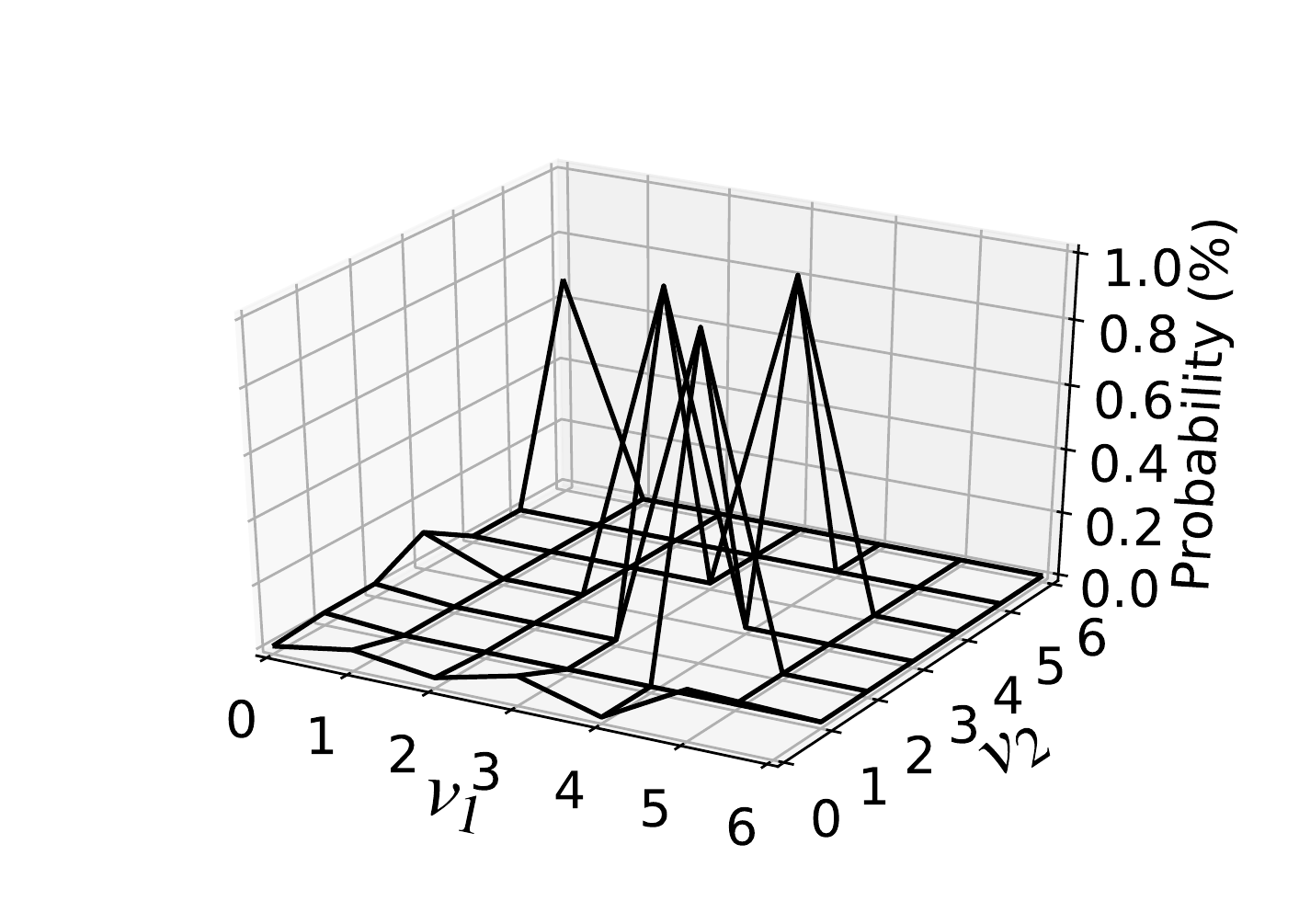} & 
\includegraphics[width=5.50cm,bb=70 0 400 250,clip]{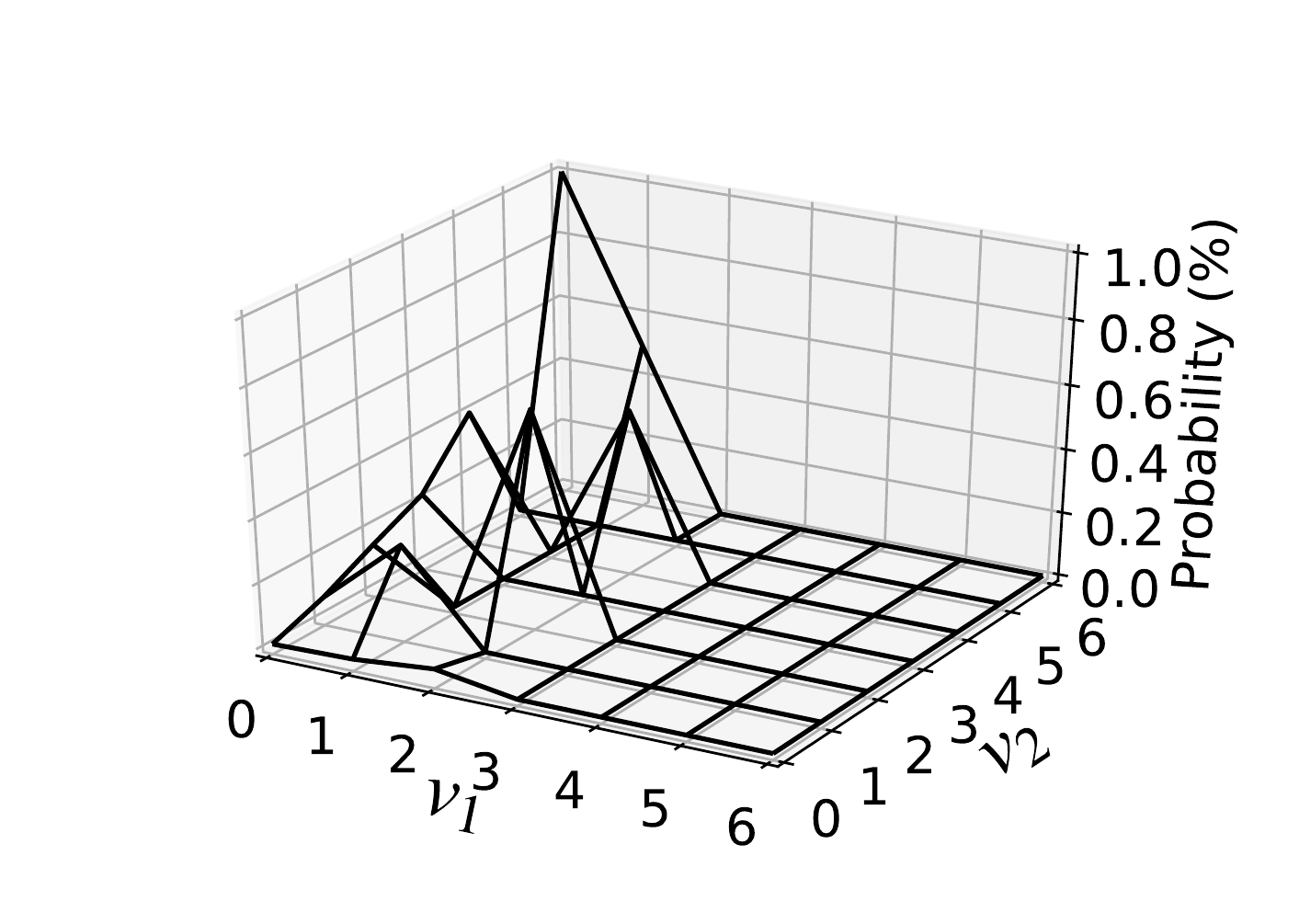}
\\[-.1in]
(a) SeqEmp, $v_3=0$ & 
(b) SeqEmp, $v_3=1$ & 
(c) SeqEmp, $v_3=2$
\\[.1in]
\includegraphics[width=5.50cm,bb=70 0 400 250,clip]{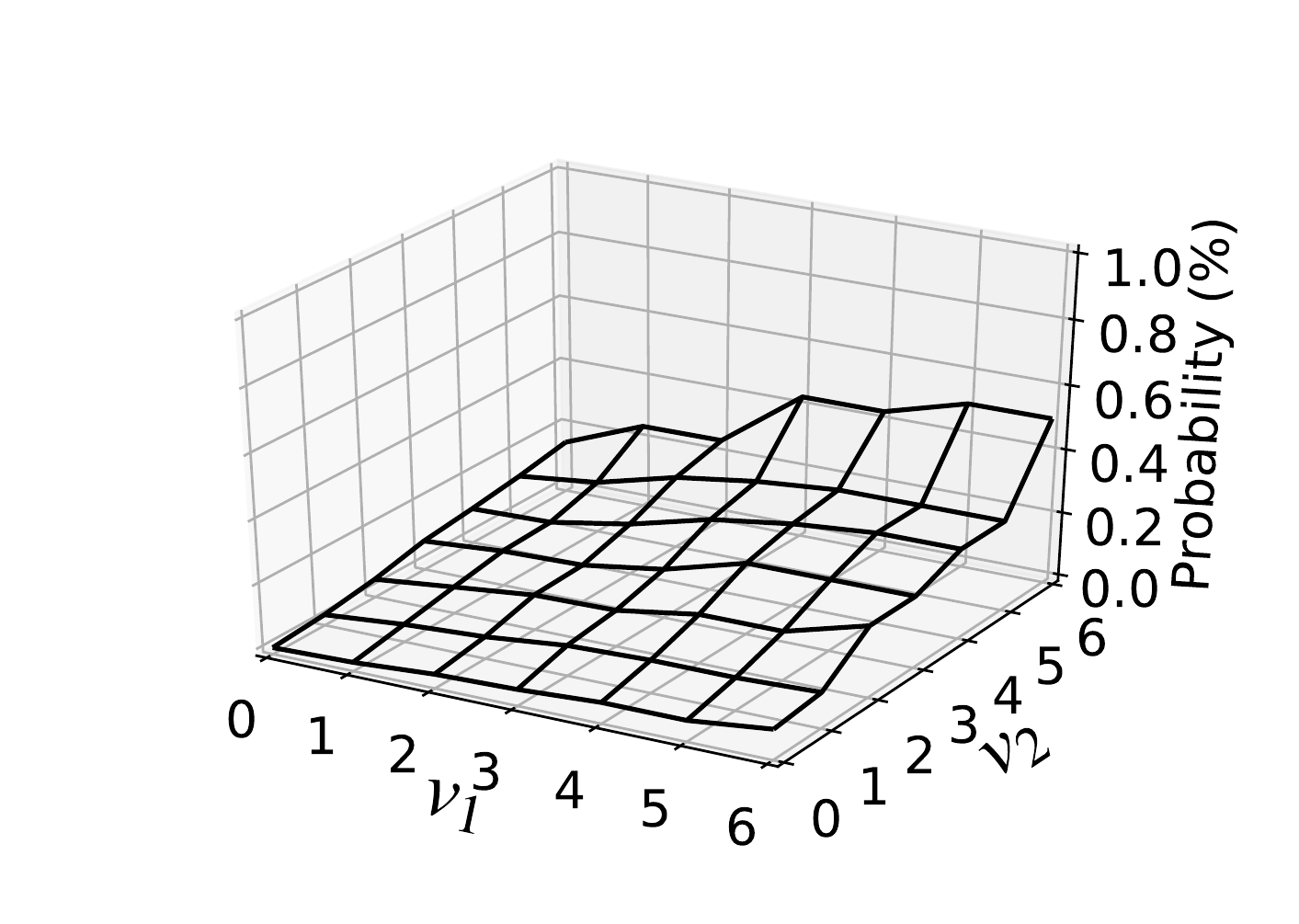} &
\includegraphics[width=5.50cm,bb=70 0 400 250,clip]{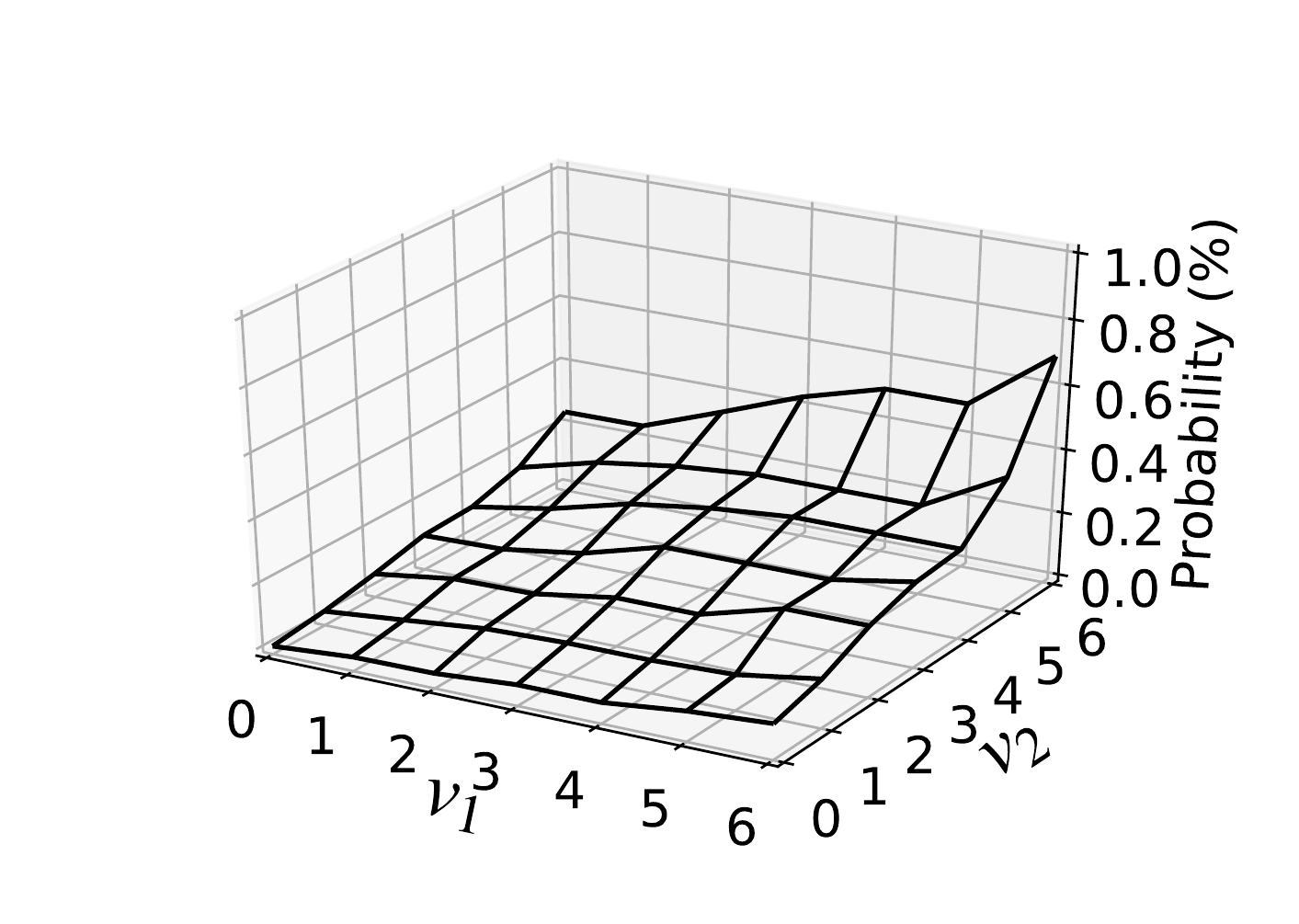} &
\includegraphics[width=5.50cm,bb=70 0 400 250,clip]{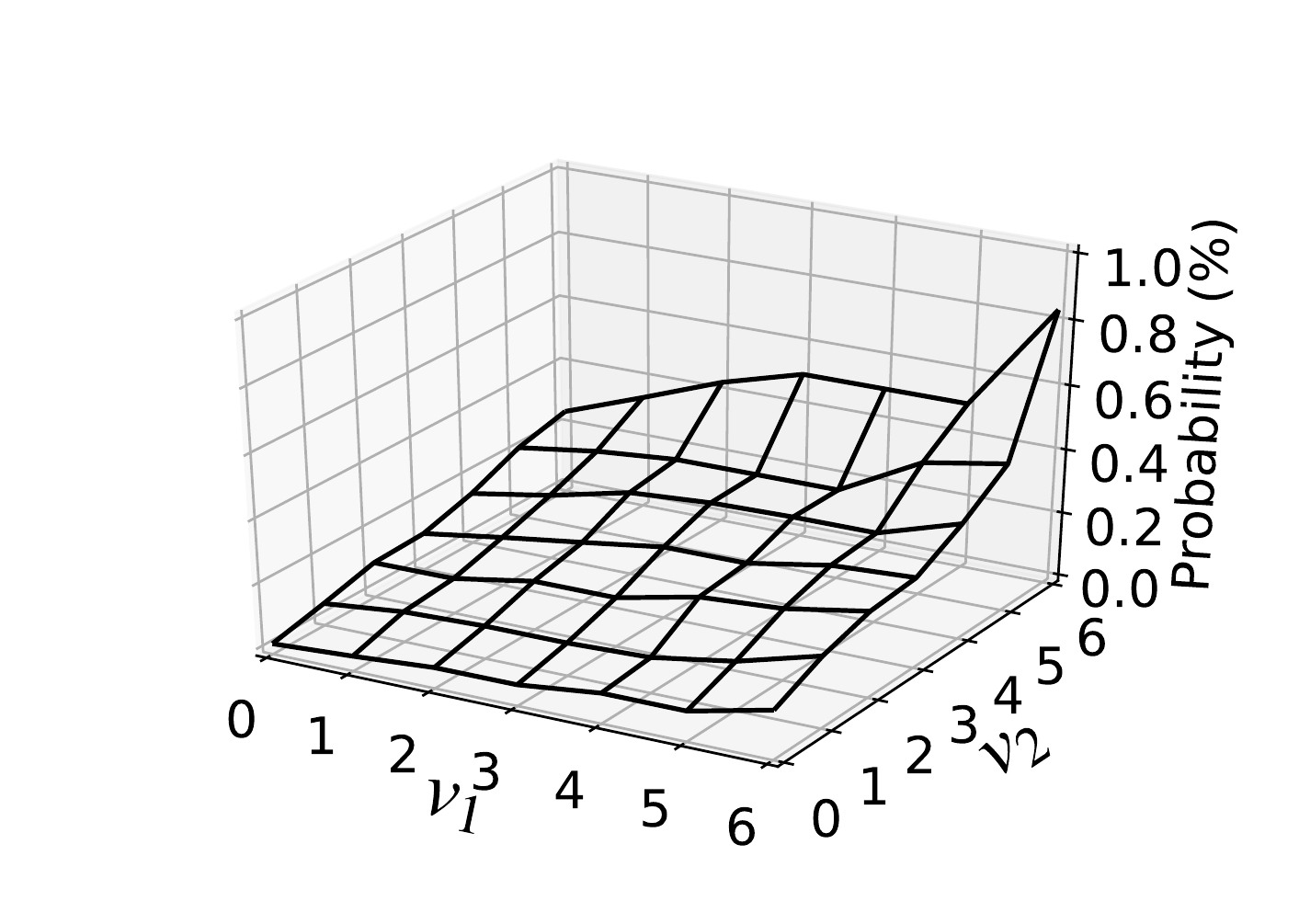}
\\[-.1in]
(d) SeqUM, $v_3=0$ & 
(e) SeqUM, $v_3=1$ & 
(f) SeqUM, $v_3=2$
\\[.1in]
\includegraphics[width=5.50cm,bb=70 0 400 250,clip]{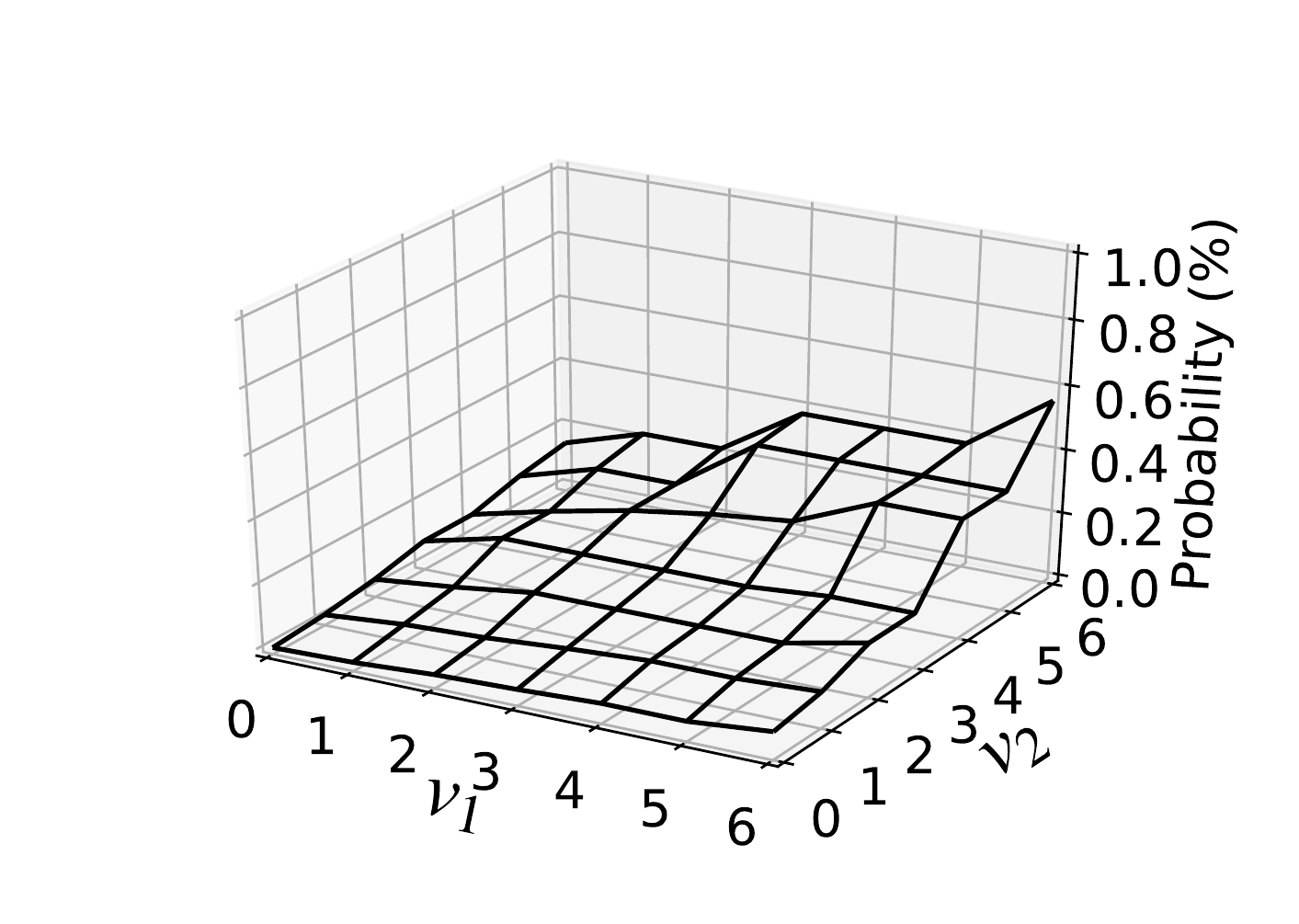} & 
\includegraphics[width=5.50cm,bb=70 0 400 250,clip]{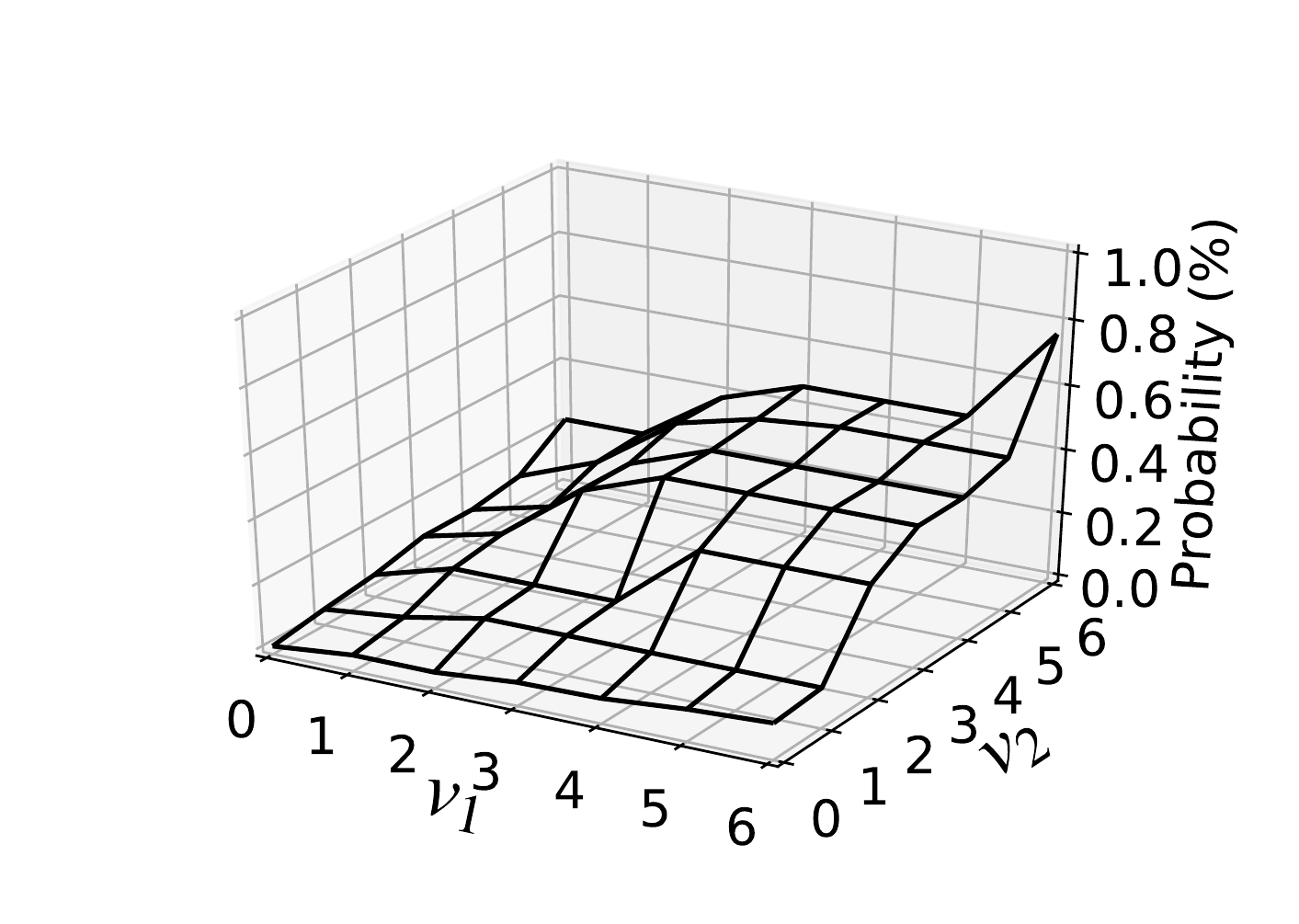} & 
\includegraphics[width=5.50cm,bb=70 0 400 250,clip]{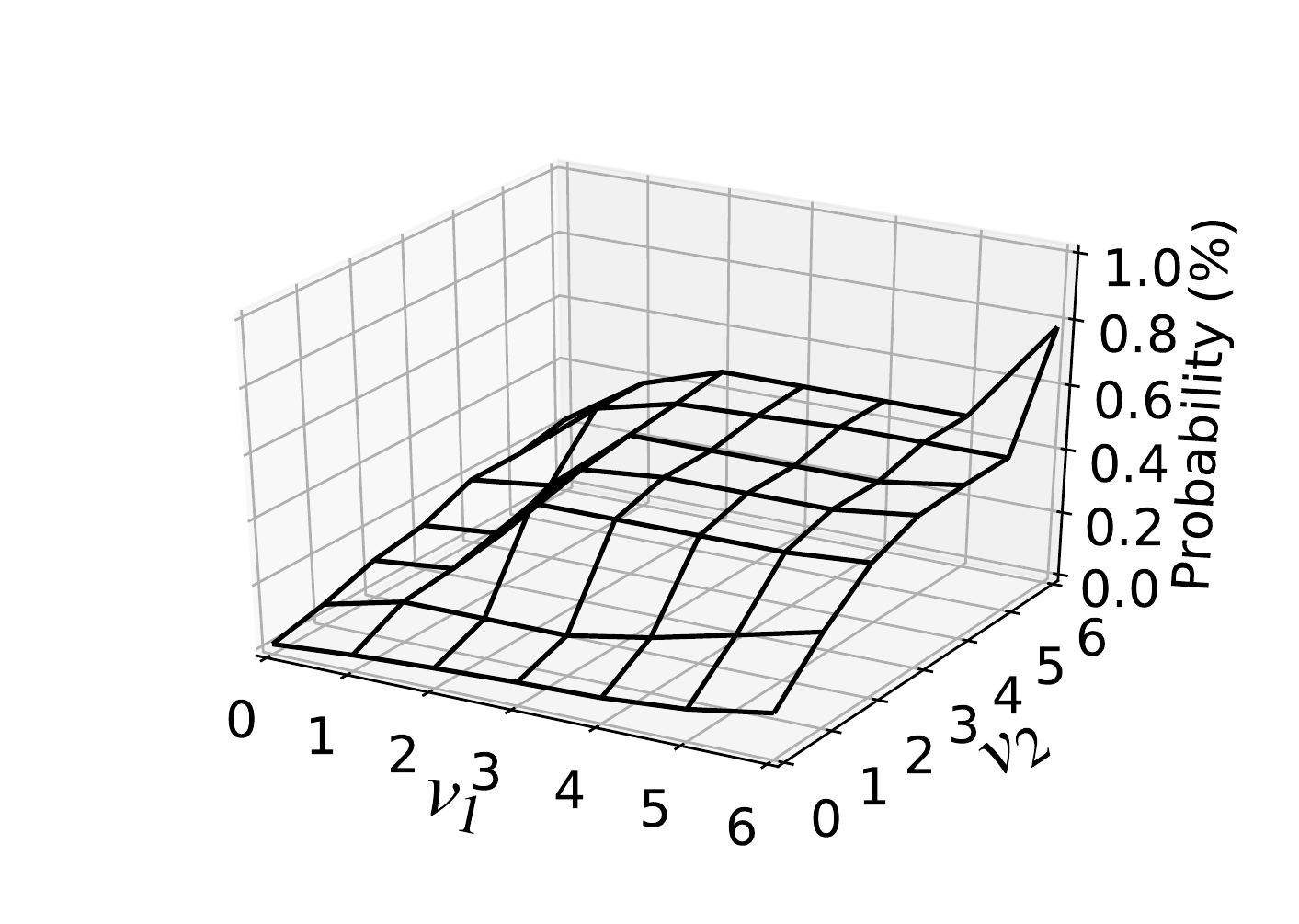}
\\[-.1in]
(g) SeqUS, $v_3=0$ & 
(h) SeqUS, $v_3=1$ & 
(i) SeqUS, $v_3=2$
\end{tabular}
\caption{Item-choice probabilities estimated from the 10\%-sampled training set with $(n,m)=(5,6)$. \label{fig:seq_visualize_sampled}}
\end{figure*}


\section{Conclusion}
\label{sec:7}

We presented a shape-restricted optimization model for estimating item-choice probabilities on an e-commerce website. 
Our monotonicity constraints based on tailored order relations could better estimate item-choice probabilities for all possible PV sequences. 
To improve computational efficiency of our optimization model, we devised constructive algorithms for transitive reduction that remove all redundant constraints from the optimization model. 

We assessed the effectiveness of our method through experiments using real-world clickstream data. 
Experimental results demonstrated that transitive reduction enhanced the efficiency of our optimization model in terms of both computation time and memory usage. 
In addition, our method delivered better prediction performance than did the two-dimensional monotonicity model~\cite{IwNi16} and common machine learning methods. 
Our method was also helpful in correcting prediction values computed by other machine learning methods. 

This study made three main contributions.
First, we derived two types of posets by exploiting the properties of recency and frequency of a user's previous PVs. 
These posets allow us to place appropriate monotonicity constraints on item-choice probabilities.  
Next, we developed algorithms for transitive reduction of our posets. 
These algorithms are more efficient than general-purpose algorithms in terms of time complexity for transitive reduction. 
Finally, our method further expanded the potential of shape-restricted regression for predicting user behavior on e-commerce websites. 

Once the optimization model for estimating item-choice probabilities has been solved, the obtained results can easily be put into practical use on e-commerce websites. 
Accurate estimates of item-choice probabilities will be useful for customizing sales promotions according to the needs of a particular user. 
In addition, our method can estimate user preferences from clickstream data, therefore aiding creation of high-quality user--item rating matrices for recommendation algorithms~\cite{IwNi19}. 

In future studies, we will develop new posets that further improve the prediction performance of our PV sequence model.  
Another direction of future research will be to incorporate user--item heterogeneity into our optimization model, as in the case of latent class modeling with a two-dimensional probability table~\cite{NiSu18}.

\appendices
\section{Proofs}
\label{sec:8}

\subsection{Proof of Theorem~\ref{thm:iffUM}}
\label{app:a1}

\subsubsection*{The ``only if'' part}
Suppose $(\bm{u},\bm{v}) \in E^*_{\texttt{UM}}$. 
We then have $\bm{v} \in \texttt{UM}(\{\bm{u}\})$ from Definition~\ref{def:UM} and Lemma~\ref{lem:rdc}. 
We therefore consider the following two cases: 

\noindent
\underline{Case 1: $\bm{v} = \texttt{Up}(\bm{u},s)$ for some $s \in [1,n]$}\\[.03in]
For the sake of contradiction, assume that $s \not= n$ (i.e., $s \le n-1$).   
Then there exists an index $j$ such that $s < j \le n$. 
If $u_j > 0$, then $\bm{w} = \texttt{Move}(\bm{u},s,j)$ and $\bm{v} = \texttt{Up}(\bm{w},j)$. 
If $u_j = 0$, then $\bm{w} = \texttt{Up}(\bm{u},j)$ and $\bm{v} = \texttt{Move}(\bm{w},s,j)$. 
This implies that $\bm{u} \prec_{\texttt{UM}} \bm{w} \prec_{\texttt{UM}} \bm{v}$, which contradicts $(\bm{u},\bm{v}) \in E^*_{\texttt{UM}}$ due to condition~(C2) of Lemma~\ref{lem:rdc}. 

\noindent
\underline{Case 2: $\bm{v} = \texttt{Move}(\bm{u},s,t)$ for some $(s,t) \in [1,n] \times [1,n]$}\\[.03in]
For the sake of contradiction, assume that $t \not= s+1$ (i.e., $t \ge s+2$). 
Then there exists an index $j$ such that $s < j < t$. 
If $u_j > 0$, then $\bm{w} = \texttt{Move}(\bm{u},s,j)$ and $\bm{v} = \texttt{Move}(\bm{w},j,t)$. 
If $u_j = 0$, then $\bm{w} = \texttt{Move}(\bm{u},j,t)$ and $\bm{v} = \texttt{Move}(\bm{w},s,j)$. 
This implies that $\bm{u} \prec_{\texttt{UM}} \bm{w} \prec_{\texttt{UM}} \bm{v}$, which contradicts $(\bm{u},\bm{v}) \in E^*_{\texttt{UM}}$ due to condition~(C2) of Lemma~\ref{lem:rdc}. 

\subsubsection*{The ``if'' part}
Next, we show that $(\bm{u},\bm{v}) \in E^*_{\texttt{UM}}$ in the following two cases:

\noindent
\underline{Case 1: Condition~(UM1) is fulfilled}\\[.03in]
Condition~(C1) of Lemma~\ref{lem:rdc} is clearly satisfied. 
To satisfy condition~(C2), we consider $\bm{w} \in \Gamma$ such that $\bm{u} \preceq_{\texttt{UM}} \bm{w} \preceq_{\texttt{UM}} \bm{v}$. 
From Lemma~\ref{lem:lex}, we have $\bm{u} \preceq_{\texttt{lex}} \bm{w} \preceq_{\texttt{lex}} \bm{v}$. 
Since $\bm{u}$ is next to $\bm{v}$ in the lexicographic order, we have $\bm{w} \in \{\bm{u},\bm{v}\}$.

\noindent
\underline{Case 2: Condition~(UM2) is fulfilled}\\[.03in]
Condition~(C1) of Lemma~\ref{lem:rdc} is clearly satisfied. 
To satisfy condition~(C2), we consider $\bm{w} \in \Gamma$ such that $\bm{u} \preceq_{\texttt{UM}} \bm{w} \preceq_{\texttt{UM}} \bm{v}$. 
From Lemma~\ref{lem:lex}, we have $\bm{u} \preceq_{\texttt{lex}} \bm{w} \preceq_{\texttt{lex}} \bm{v}$, which implies that $w_j = u_j$ for all $j \in [1,s-1]$. 
Therefore, we cannot apply any operations to $w_j$ for $j \in [1,s-1]$ in the process of transforming $\bm{w}$ from $\bm{u}$ into $\bm{v}$.
To keep the value of $\sum_{j=1}^n w_j$ constant, we can apply only the \texttt{Move} operation. 
However, once the \texttt{Move} operation is applied to $w_j$ for $j \in [s+2,n]$, the resultant sequence cannot be converted into $\bm{v}$. 
As a result, only $\texttt{Move}(\,\cdot\,,s,s+1)$ can be performed, 
and therefore $\bm{w} = \bm{u}$ or $\bm{w} = \texttt{Move}(\bm{u},s,s+1) = \bm{v}$. 

\subsection{Proof of Theorem~\ref{thm:iffUS}}
\label{app:a2}

\subsubsection*{The ``only if'' part}
Suppose that $(\bm{u},\bm{v}) \in E^*_{\texttt{US}}$. 
We then have $\bm{v} \in \texttt{US}(\{\bm{u}\})$ from Definition~\ref{def:US} and Lemma~\ref{lem:rdc}. 
Thus, we consider the following two cases: 

\noindent
\underline{Case 1: $\bm{v} = \texttt{Up}(\bm{u},s)$ for some $s \in [1,n]$}\\[.03in]
For the sake of contradiction, assume $u_j \in \{u_s, u_s + 1\}$ for some $j \in [s+1,n]$. 
If $u_j = u_s$, then $\bm{w} = \texttt{Up}(\bm{u},j)$ and $\bm{v} = \texttt{Swap}(\bm{w},s,j)$. 
If $u_j = u_s + 1$, then $\bm{w} = \texttt{Swap}(\bm{u},s,j)$ and $\bm{v} = \texttt{Up}(\bm{w},j)$. 
This implies that $\bm{u} \prec_{\texttt{US}} \bm{w} \prec_{\texttt{US}} \bm{v}$, which contradicts $(\bm{u},\bm{v}) \in E^*_{\texttt{US}}$ due to condition~(C2) of Lemma~\ref{lem:rdc}. 

\noindent
\underline{Case 2: $\bm{v} = \texttt{Swap}(\bm{u},s,t)$ for some $(s,t) \in [1,n] \times [1,n]$}\\[.03in]
For the sake of contradiction, assume $u_j \in [u_s, u_t]$ for some $j \in [s+1,t-1]$. 
If $u_s < u_j < u_t$, then $\bm{w}_1 = \texttt{Swap}(\bm{u},j,t)$, $\bm{w}_2 = \texttt{Swap}(\bm{w}_1,s,j)$, and $\bm{v} = \texttt{Swap}(\bm{w}_2,j,t)$. 
If $u_j = u_s$, then $\bm{w} = \texttt{Swap}(\bm{u},j,t)$ and $\bm{v} = \texttt{Swap}(\bm{w},s,j)$. 
If $u_j = u_t$, then $\bm{w} = \texttt{Swap}(\bm{u},s,j)$ and $\bm{v} = \texttt{Swap}(\bm{w},j,t)$.
Each of these results contradicts $(\bm{u},\bm{v}) \in E^*_{\texttt{US}}$ due to condition~(C2) of Lemma~\ref{lem:rdc}. 

\subsubsection*{The ``if'' part}
Next, we show that $(\bm{u},\bm{v}) \in E^*_{\texttt{US}}$ in the following two cases:

\noindent
\underline{Case 1: Condition~(US1) is fulfilled}\\[.03in]
Condition~(C1) of Lemma~\ref{lem:rdc} is clearly satisfied. 
To satisfy condition~(C2), we consider $\bm{w} \in \Gamma$ such that $\bm{u} \preceq_{\texttt{US}} \bm{w} \preceq_{\texttt{US}} \bm{v}$. 
From Lemma~\ref{lem:lex}, we have $\bm{u} \preceq_{\texttt{lex}} \bm{w} \preceq_{\texttt{lex}} \bm{v}$, implying that $w_j = u_j$ for all $j \in [1,s-1]$. 
Therefore, we cannot apply any operations to $w_j$ for $j \in [1,s-1]$ in the process of transforming $\bm{w}$ from $\bm{u}$ into $\bm{v}$. 
We must apply the \texttt{Up} operation only once, because the value of $\sum_{j=1}^n w_j$ remains the same after the \texttt{Swap} operation. 
Condition~(US1) guarantees that for all $j \in [s+1,n]$, $w_j$ does not coincide with $u_s + 1$ even if $\texttt{Up}(\,\cdot\,,j)$ is applied. 
Therefore, $\texttt{Swap}(\,\cdot\,,s,j)$ for $j \in [s+1,n]$ never leads to $w_s = u_s + 1$. 
As a result, $\texttt{Up}(\,\cdot\,,s)$ must be performed.  
Other applicable \texttt{Swap} operations produce a sequence that cannot be converted into $\bm{v}$.
This means that $\bm{w} = \bm{u}$ or $\bm{w} = \texttt{Up}(\bm{u},s) = \bm{v}$. 

\noindent
\underline{Case 2: Condition~(US2) is fulfilled}\\[.03in]
Condition~(C1) of Lemma~\ref{lem:rdc} is clearly satisfied. 
To satisfy condition~(C2), we consider $\bm{w} \in \Gamma$ such that $\bm{u} \preceq_{\texttt{US}} \bm{w} \preceq_{\texttt{US}} \bm{v}$. 
From Lemma~\ref{lem:lex}, we have $\bm{u} \preceq_{\texttt{lex}} \bm{w} \preceq_{\texttt{lex}} \bm{v}$. 
This implies that $w_j = u_j$ for all $j \in [1,s-1]$, and that $w_s \in [u_s, u_t]$. 
Therefore, we cannot apply any operations to $w_j$ for $j \in [1,s-1]$ in the process of transforming $\bm{w}$ from $\bm{u}$ into $\bm{v}$.
To keep the value of $\sum_{j=1}^n w_j$ constant, we can apply only the \texttt{Swap} operation. 
However, once the \texttt{Swap} operation is applied to $w_j$ for $j \in [t+1,n]$, the resultant sequence cannot be converted into $\bm{v}$. 
We cannot adopt $\bm{w} = \texttt{Swap}(\bm{u},s,j)$ for $j \in [s+1,t-1]$ due to condition~(US2). 
If we adopt $\bm{w} = \texttt{Swap}(\bm{u},j,t)$ for $j \in [s+1,t-1]$, we have $w_t \le u_s - 1$ due to condition~(US2), so the application of $\texttt{Swap}(\,\cdot\,,t,j)$ is unavoidable for $j \in [t+1,n]$. 
As a result, $\texttt{Swap}(\,\cdot\,,s,t)$ must be performed. 
Other applicable \texttt{Swap} operations produce a sequence that cannot be converted into $\bm{v}$.
This means that $\bm{w} = \bm{u}$ or $\bm{w} = \texttt{Swap}(\bm{u},s,t) = \bm{v}$. 

\section{Pseudocode}\label{app1}

\subsection{Constructive algorithm for $(\Gamma, E^*_{\normalfont \texttt{UM}})$}
\label{app1um}

\begin{algorithm}[t]
\caption{Constructive algorithm for $(\Gamma, E^*_{\texttt{UM}})$}\label{app:b1}
\hspace*{\algorithmicindent} \textbf{Input} a pair $(n,m)$ of positive integers\\
\hspace*{\algorithmicindent} \textbf{Output} transitive reduction $(\Gamma, E^*_{\texttt{UM}})$
\begin{algorithmic}[1]
\Procedure{}{}
\State $L \gets \mbox{list consisting of }(0,0,\ldots,0)$ \Comment{returns $\Gamma$}
\State $E \gets \mbox{empty list}$ \Comment{returns $E^*_{\texttt{UM}}$}
\State $Q \gets \mbox{queue consisting of }(0,0,\ldots,0)$
\While{$Q$ is not empty}
	\State $\bm{u} \gets \textsc{dequeue}(Q)$
	\If{$(\bm{u},n) \in \mathcal{D}_{\texttt{U}}$} \Comment{for (UM1)}
		\State $\bm{v} \gets \texttt{Up}(\bm{u},n)$	
		\State $\textsc{append}(L,\bm{v})$, $\textsc{append}(E,(\bm{u},\bm{v}))$
		\State $\textsc{enqueue}(Q,\bm{v})$
	\EndIf
	\For{$s \in [1,n-1]$}\Comment{for (UM2)}
		\If{$(\bm{u},s,s+1) \in \mathcal{D}_{\texttt{M}}$} 
			\State $\bm{v} \gets \texttt{Move}(\bm{u},s,s+1)$		
			\State $\textsc{append}(L,\bm{v})$, $\textsc{append}(E,(\bm{u},\bm{v}))$
			\State $\textsc{enqueue}(Q,\bm{v})$
		\EndIf
	\EndFor
\EndWhile
\EndProcedure
\end{algorithmic}
\end{algorithm}

The nodes and directed edges of graph $(\Gamma, E^*_{\texttt{UM}})$ are enumerated in a breadth-first search 
and are stored in two lists $L$ and $E$, respectively. 
We use $\textsc{append}(L,\bm{v})$, which appends a vertex $\bm{v}$ to the end of $L$. 
We similarly use $\textsc{append}(E,(\bm{u},\bm{v}))$. 

A queue $Q$ is used to store nodes of $L$ whose successors are under investigation (the ``frontier'' of $L$). 
The nodes in $Q$ are listed in ascending order of depth, where 
the depth of $\bm{v}$ is the shortest-path length from $(0,0,\ldots,0)$ to $\bm{v}$. 
We use $\textsc{dequeue}(Q)$, which returns and deletes the first element in $Q$, 
and $\textsc{enqueue}(Q,\bm{v})$, which appends $\bm{v}$ to the end of $Q$. 

Algorithm~\ref{app:b1} summarizes our constructive algorithm for computing the transitive reduction $(\Gamma, E^*_{\texttt{UM}})$. 
For a given node $\bm{u}$ in line 6, we find all nodes $\bm{v}$ satisfying condition (UM1) in lines 7--10 and those satisfying condition (UM2) in lines 11--15.

By definition, the membership test for $\mathcal{D}_{\texttt{U}}$ and $\mathcal{D}_{\texttt{M}}$ can be performed in $\mathcal{O}(1)$ time. 
Recall that $\textsc{dequeue}$, $\textsc{enqueue}$, and $\textsc{append}$ can be performed in $\mathcal{O}(1)$ time. 
The FOR loop in lines 11--15 executes in $\mathcal{O}(n)$ time. 
Therefore, recalling that $|\Gamma|=(m+1)^n$, we see that Algorithm~\ref{app:b1} runs in $\mathcal{O}(n (m+1)^n)$ time. 

\subsection{Constructive algorithm for $(\Gamma, E^*_{\normalfont \texttt{US}})$}
\label{app1us}
Algorithm~\ref{app:b2} summarizes our constructive algorithm for computing the transitive reduction $(\Gamma, E^*_{\texttt{US}})$. 
Here, the difference from Algorithm~\ref{app:b1} is the method for finding nodes $\bm{v}$ satisfying conditions (US1) or (US2). 
For a given node $\bm{u}$ in line 6, we find all nodes $\bm{v}$ satisfying condition (US1) in lines 7--16, 
and those satisfying condition (US2) in lines 17--26. 
The following describes the latter part. 

\begin{algorithm}[t]
\caption{Constructive algorithm for $(\Gamma, E^*_{\texttt{US}})$}\label{app:b2}
\hspace*{\algorithmicindent} \textbf{Input:} a pair $(n,m)$ of positive integers\\
\hspace*{\algorithmicindent} \textbf{Output:} the transitive reduction $(\Gamma, E^*_{\texttt{US}})$
\begin{algorithmic}[1]
\Procedure{}{}
\State $L \gets \mbox{list consisting of }(0,0,\ldots,0)$ \Comment{returns $\Gamma$}
\State $E \gets \mbox{empty list}$ \Comment{returns $E^*_{\texttt{US}}$}
\State $Q \gets \mbox{queue consisting of }(0,0,\ldots,0)$
\While{$Q$ is not empty}
	\State $\bm{u} \gets \textsc{dequeue}(Q)$
	\For{$s \in [1,n]$} \Comment{for (US1)}
		\If{$(\bm{u},s) \in \mathcal{D}_{\texttt{U}}$}
			\State $\textit{flag} \gets \textit{True}$
			\For{$j \in [s+1,n]$}
				\If{$u_j \in \{ u_s, u_s+1\}$}
					\State $\textit{flag} \gets \textit{False}$, {\bf break}
				\EndIf
			\EndFor
		\If{$\textit{flag} = \textit{True}$}
			\State $\bm{v} \gets \texttt{Up}(\bm{u},s)$
			\State $\textsc{append}(L,\bm{v})$, $\textsc{append}(E,(\bm{u},\bm{v}))$
			\State $\textsc{enqueue}(Q,\bm{v})$
		\EndIf
		\EndIf
	\EndFor
	\For{$s \in [1,n-1]$} \Comment{for (US2)}
		\State $b \gets m+1$ 
		\For{$t \in [s+1,n]$}
			\If{$(\bm{u},s,t) \in \mathcal{D}_{\texttt{S}}$ {\bf and} $u_t < b$}
				\State $\bm{v} \gets \texttt{Swap}(\bm{u},s,t)$
				\State $\textsc{append}(L,\bm{v})$, $\textsc{append}(E,(\bm{u},\bm{v}))$
				\State $\textsc{enqueue}(Q,\bm{v})$
				\State $b \gets u_t$
			\ElsIf{$u_t = u_s$}
				\State {\bf break}
			\EndIf
		\EndFor
	\EndFor
\EndWhile
\EndProcedure
\end{algorithmic}
\end{algorithm}

Let $(\bm{u},\bm{v})$ be a directed edge added to $E$ in line 22. 
Let $(\bar{s},\bar{t})$ be such that $\bm{v} = \texttt{Swap}(\bm{u},\bar{s},\bar{t})$. 
From line 20, we have 
$u_{\bar{s}} < u_{\bar{t}} < b$. 
Note that 
for each $t$ in line 19, 
value $b$ gives the smallest value of $u_j$ with $u_j > u_{\bar{s}}$ for $j\in [\bar{s}+1,t-1]$. 
Also, due to lines 25--26, $u_j \not= u_{\bar{s}}$ for $j\in [\bar{s}+1,\bar{t}-1]$. 
Combining these observations, we see that for $j\in [\bar{s}+1,\bar{t}-1]$, 
\[
u_j < u_{\bar{s}}~~\mbox{or}~~ u_j \ge b > u_{\bar{t}} ~~ \mbox{(meaning $u_j \notin [u_{\bar{s}}, u_{\bar{t}}]$)}. 
\]
Therefore, the pair $(\bm{u},\bm{v})$ satisfies condition (US2). 
It is easy to verify that this process finds all vertices $\bm{v}$ satisfying condition (US2). 

Since both of the double FOR loops at lines 7--16 and 17--26 execute in $\mathcal{O}(n^2)$ time,  
Algorithm~\ref{app:b2} runs in $\mathcal{O}(n^2 (m+1)^n)$ time. 

\end{document}